\documentclass[11pt]{article}
\usepackage{graphicx, float} 
\usepackage{amsmath,amssymb,amsthm}
\usepackage{cite}
\usepackage{url}
\usepackage{hyperref}
\usepackage{bm}
\usepackage{physics}
\usepackage{microtype}
\usepackage{setspace}
\usepackage{abstract}
\hypersetup{
    colorlinks=true,
    linkcolor=blue,
    citecolor=blue,
    urlcolor=blue
}

\usepackage{graphicx,epstopdf} 
\usepackage[caption=false]{subfig} 
\usepackage[top=1in, bottom=1in, left=1in, right=1in]{geometry}

\newtheorem{prop}{Proposition}[subsection]

\title{A quantum model of opinion dynamics on networks}

\author{Weiqi Chu} 
\date{Department of Mathematics and Statistics \\ University of Massachusetts Amherst}

\begin{document}

\maketitle

\begin{abstract}
Classical models of opinion dynamics represent individual opinions as scalar or vector values governed by the classical probability theory, either as deterministic quantities or random variables. This framework does not account for empirically observed phenomena such as cognitive ambivalence (where an individual simultaneously holds conflicting views) and order effects (where survey responses depend on the order in which questions are asked). We propose a quantum model of opinion dynamics in which each agent's cognitive state is represented by a density matrix that encodes both the expressed opinion and cognitive ambivalence. Survey questions become non-commuting self-adjoint operators, which provides a principled explanation for order effects. Our model also identifies quantities without classical counterparts, including quantum coherence and pairwise opinion covariances. Under a product state approximation, the quantum model reduces to the classical Friedkin--Johnsen opinion model. We test the framework on synthetic and real-world networks and observe that pairwise correlations follow network-dependent transient dynamics but converge to the same steady state regardless of the network, and that quantum coherence decays exponentially at a rate independent of the network.
\end{abstract}

\section{Introduction}
Opinion dynamics studies how individual beliefs evolve and form through social interaction~\cite{castellano2009statistical}. Classical opinion models represent each agent's opinion as a scalar or vector and prescribe rules on how agents update their opinions based on interactions with their neighbors~\cite{li2026bounded,sirbu2016opinion}. The DeGroot model~\cite{degroot1974reaching} updates agents' opinions with a weighted average over their neighbors, and the Friedkin--Johnsen (FJ) model~\cite{friedkin1990social} extends the DeGroot model by introducing a bias toward each agent's anchor opinion. In addition, bounded-confidence models such as the Deffuant--Weisbuch and Hegselmann--Krause models~\cite{deffuant2000mixing,hegselmann2002opinion} allow interactions only between agents whose opinion differences are less than a prescribed threshold. Another class of models considers the inherent random effects in human interactions. The classical voter model~\cite{holley1975ergodic} updates each agent by adopting the opinion of a randomly chosen neighbor, and its noisy variants introduce spontaneous flips to account for external environment influence~\cite{carro2016noisy}. Social contagion models view opinion adoption as an infection process, analogous to epidemic spreading~\cite{castellano2009statistical}.

Despite their success, both deterministic and random models assume the opinions and viewpoints are definite at any given time, obeying the Kolmogorov axioms of classical probability. This assumption becomes limiting when modeling phenomena that involve human judgment and cognition. One such phenomenon is the order effect~\cite{moore2002measuring}, where survey respondents give different answers to the same question, such as approval ratings for political figures, when questions are asked in different orders. Another example is the disjunction fallacy~\cite{busemeyer2012quantum}, where people rate the likelihood of two events occurring together as lower than the likelihood of either event alone, violating the classical probability addition rule. These phenomena are difficult to reconcile with classical opinion models, and quantum cognition addresses these limitations by replacing the classical Kolmogorov axioms with the non-commutative probability framework of quantum mechanics~\cite{busemeyer2012quantum,pothos2022quantum}, in which belief states are represented as unit vectors in a Hilbert space and modeling expressed opinions as projective measurements. Busemeyer et al. explain order effects and cognitive fallacies from the non-commutativity of measurement operators~\cite{busemeyer2012quantum} and make parameter-free predictions on the structure of order effects that have been confirmed across large survey datasets~\cite{wang2013quantum,pothos2022quantum}.

Recent efforts have extended the quantum cognition framework to multi-agent dynamics on networks~\cite{guo2025quantum,bottcher2024complex,sans2026modeling}, in which the opinion evolution is governed by unitary or Hamiltonian dynamics, which are time-reversible and energy-conserving. However, social interactions and influence are inherently irreversible, where repeated interactions reduce cognitive ambivalence and enhance opinion crystallization over time~\cite{zaller1992nature,luttrell2016ambivalence}. Moreover, most existing quantum opinion models do not establish a formal connection to classical models. A dissipative quantum framework is therefore needed to explain cognitive phenomena such as order effects and ambivalence, and to establish a formal connection to classical opinion models.

The Lindblad master equation~\cite{lindblad1976generators,gorini1976completely} models open quantum systems by describing irreversible dissipative interactions via jump operators. It has found broad success in quantum optics, condensed matter physics, and quantum information~\cite{breuer2002theory,nielsen2010quantum}, and has recently attracted attention in cognitive science and the social sciences~\cite{khrennikov2023open}.
Asano et al.~\cite{asano2011application} introduced the Lindblad equation for belief updating in cognitive psychology and showed that its equilibrium solution explains the disjunction effect. Khrennikova et al.~\cite{khrennikova2014application} applied it to political party governance and voting behavior. In game theory, Asano and Khrennikov~\cite{asano2026quantum} showed that dissipative quantum evolution stabilizes non-Nash outcomes such as those in the Prisoner's Dilemma. These works, however, focus on single-agent decision making and do not address multi-agent interactions, collective behaviors, or influence of network structure.

We model opinion dynamics on social networks using the Lindblad framework, representing each agent's cognitive state by a density matrix that encodes expressed opinion and cognitive ambivalence. Opinions arise naturally as expectation values of measurement observables under this framework. Survey questions act as self-adjoint operators whose non-commutativity gives a principled account of order effects. We also derive governing equations for the single-agent marginal density matrix and show that projecting onto the opinion observable exactly recovers the classical FJ model. Our model gives access to quantities without classical counterparts, such as coherence magnitude and phase. We test our framework numerically on synthetic networks and a real-world social network from the Facebook-100 dataset, comparing the full quantum model with its classical reduction and examining the role of quantum coherence in opinion formation.

We organize the paper as follows. In Section~\ref{sec:representation}, we introduce the quantum representation of cognitive states and the opinion observable, and establish a framework for measuring order effects. Section~\ref{sec:models} introduces the Lindblad master equation for opinion dynamics and derives the reduced single-agent dynamics under the product state approximation. Once projected onto the opinion observable, the master equation recovers the continuous-time FJ opinion model. In Section~\ref{sec:numerics}, we present numerical results on synthetic and real-world networks. We conclude in Section~\ref{sec:conclusion}. Appendices~\ref{app:marginal} and~\ref{app:coherence} contain the derivations of the reduced master equation and the coherence dynamics. Our code is available at \href{https://github.com/weiqichu/quantum\_opinion\_model\_Lindblad\_FJ}{https://github.com/weiqichu/quantum\_opinion\_model\_Lindblad\_FJ}.

\section{A quantum representation of opinions}
\label{sec:representation}
\subsection{The density matrix and opinion observable}

Let $\ket{0} = (1,0)^\top$ and $\ket{1} = (0,1)^\top$ be the two basis states in $\mathbb{C}^2$. Each agent's belief state is a unit-length vector in $\mathbb{C}^2$~\cite{busemeyer2012quantum}, represented as a superposition
\begin{equation}
\ket{\psi} = \alpha\ket{0} + \beta\ket{1}, \quad |\alpha|^2 + |\beta|^2 = 1\,.
\label{eq:superposition}
\end{equation}
We define the opinion observable as
\begin{equation}
\hat{O} = \begin{pmatrix} 1 & 0 \\ 0 & -1 \end{pmatrix}\,,
\label{eq:sigmaZ}
\end{equation}
whose eigenvalues are $1$ and $-1$, with eigenstates $\ket{0}$ and $\ket{1}$, corresponding to full commitment to opinion~A and opinion~B, respectively. 
The superposition \eqref{eq:superposition} captures cognitive ambivalence~\cite{kaplan1972ambivalence}, whereby an agent entertains conflicting views without having settled on a definite opinion. 
Upon measurement, $|\alpha|^2$ and $|\beta|^2$ are the probabilities of observing opinion~A and opinion~B, respectively.

We consider a system of $n$ agents whose interactions are encoded in a directed weighted graph with edge set $E$ and influence matrix $W$, where $w_{ij}$ denotes the influence of agent $j$ on agent $i$. 
We require $W$ to be row-stochastic with no self-loops, so that $w_{ij} \geq 0$, $w_{ii} = 0$, and $\sum_j w_{ij} = 1$ for all $i$.
We model each agent $i \in \{1,\ldots,n\}$ by a density matrix $\rho_i \in \mathcal{D}(\mathcal{H}_i)$, where $\mathcal{H}_i = \mathbb{C}^2$ and $\mathcal{D}(\mathcal{H}_i)$ denotes the set of $2\times 2$ positive-semidefinite Hermitian matrices with unit trace~\cite{nielsen2010quantum}. 
We generalize the pure state $\rho_i = \ket{\psi}\bra{\psi}$ to mixed states $\rho_i = \sum_k p_k \ket{\psi_k}\bra{\psi_k}$, with $p_k \geq 0$ and $\sum_k p_k = 1$, which capture a probabilistic mixture of pure states. We describe the joint state of all $n$ agents by $\rho \in \mathcal{D}(\mathcal{H})$, where $\mathcal{H} = \bigotimes_{i=1}^n \mathcal{H}_i$ has dimension $2^n$. 
We recover the state of agent $i$ by taking the partial trace, $\rho_i = \mathrm{Tr}_{\neq i}(\rho)$.

In quantum mechanics, each observable is represented by a self-adjoint operator whose eigenvalues are the possible measurement outcomes, and the expected outcome of an observable $\hat{O}_i$ is $\langle \hat{O}_i \rangle = \mathrm{Tr}(\rho_i \hat{O}_i)$.
Recall the opinion observable $\hat{O}$ in Eq.~\eqref{eq:sigmaZ}; we obtain the opinion of agent $i$ as $x_i = \mathrm{Tr}(\rho_i \hat{O})$.
In the quantum framework, however, $\rho_i$ encodes more information than $x_i$ alone, carried by its off-diagonal elements.
The off-diagonal entry $\rho_i^{01} = \bra{0}\rho_i\ket{1}$ is the quantum coherence of agent $i$.
The quantum coherence $\rho_i^{01}$ takes the polar form
\begin{equation}
\rho_i^{01} = \frac{c_i}{2}\,e^{-i\phi_i}\,,
\label{eq:polar}
\end{equation}
which defines the coherence magnitude $c_i = 2|\rho_i^{01}| \in [0,1]$ and the coherence phase $\phi_i = -\arg(\rho_i^{01}) \in [0,2\pi)$.
When $c_i = 0$, the density matrix reduces to a classical probabilistic mixture in which the agent holds an uncertain but definite opinion, without simultaneously entertaining conflicting views.
When $c_i \neq 0$, the agent occupies a genuine quantum superposition, a property we connect to order effects in Section~\ref{sec:order_effects}.

\subsection{A quantum account of order effects}\label{sec:order_effects}

Order effects are the empirically documented phenomenon in which the answer to a question depends on the order in which questions are asked~\cite{moore2002measuring,wang2013quantum}.
Let $\hat{\Pi}_A$ and $\hat{\Pi}_B$ be two self-adjoint operators associated with questions $Q_A$ and $Q_B$, acting on $\mathcal{H}_i$, with spectral decompositions
\begin{equation}
\hat{\Pi}_{\alpha} = \sum_k \lambda_k^{\alpha} \ket{\psi_k^{\alpha}}\bra{\psi_k^{\alpha}}\,, \quad \alpha \in \{A, B\}\,,
\label{eq:spectral}
\end{equation}
where $\lambda_k^{\alpha} \in \mathbb{R}$ are the eigenvalues and $\ket{\psi_k^{\alpha}}$ are the corresponding eigenstates.
We denote the rank-1 projectors $\hat{P}_k^{\alpha} = \ket{\psi_k^{\alpha}}\bra{\psi_k^{\alpha}}$.
By the L\"{u}ders projection rule~\cite{luders2006concerning}, measuring $Q_A$ on state $\rho_i$ and obtaining outcome $\lambda_k^A$ collapses the state to
\begin{equation}
\tilde{\rho}_i^{(1)} = \frac{\hat{P}_k^A\,\rho_i\,\hat{P}_k^A}{\mathrm{Tr}(\hat{P}_k^A\,\rho_i)}\,,
\label{eq:luders}
\end{equation}
with probability $\mathrm{Tr}(\hat{P}_k^A\rho_i)$.
When $[\hat{\Pi}_A, \hat{\Pi}_B] \neq 0$, the two observables are incompatible, and the order of measurement affects subsequent probabilities, producing an order effect.

We quantify order effects in two ways.
The first measures the joint probability of outcomes $\lambda_k^A$ and $\lambda_{\ell}^B$ under both orderings.
From the L\"{u}ders rule~\eqref{eq:luders}, the joint probability of obtaining $\lambda_k^A$ then $\lambda_{\ell}^B$ is
\begin{equation}
P_{AB} = \mathrm{Tr}(\hat{P}_k^A\rho_i)\mathrm{Tr}(\hat{P}_{\ell}^B\tilde{\rho}_i^{(1)}) = \mathrm{Tr}(\hat{P}_{\ell}^B\hat{P}_k^A\rho_i\hat{P}_k^A)\,,
\label{eq:PAB}
\end{equation}
and symmetrically $P_{BA} = \mathrm{Tr}(\hat{P}_k^A\hat{P}_{\ell}^B\rho_i\hat{P}_{\ell}^B)$ for the reversed order.
We define the order-effect asymmetry as
\begin{equation}
\Delta P_i^{(1)} = P_{AB} - P_{BA} = \mathrm{Tr}(\hat{D}^{(1)}\rho_i)\,,
\label{eq:DeltaP1}
\end{equation}
where $\hat{D}^{(1)} = \hat{P}_k^A\hat{P}_{\ell}^B\hat{P}_k^A - \hat{P}_{\ell}^B\hat{P}_k^A\hat{P}_{\ell}^B$ is a Hermitian operator determined solely by the question pair.
This formulation requires recording outcomes of both questions under both orderings, corresponding to a split-sample survey design~\cite{wang2013quantum}.

The second compares the marginal probability of obtaining $\lambda_k^A$ when $Q_A$ is asked first against when it follows a non-selective measurement of $Q_B$.
When $Q_A$ is asked first, the marginal probability of obtaining $\lambda_k^A$ is $P_A = \mathrm{Tr}(\hat{P}_k^A\rho_i)$.
A non-selective measurement of $Q_B$ corresponds to measuring $Q_B$ and discarding the outcome, updating the state to
\begin{equation}
\tilde{\rho}_i^{(2)} = \sum_{\ell} \hat{P}_{\ell}^B\rho_i\hat{P}_{\ell}^B\,,
\label{eq:nonselective}
\end{equation}
which is a classical mixture of post-measurement states in the eigenbasis of $\hat{\Pi}_B$.
The marginal probability of subsequently obtaining $\lambda_k^A$ is then
\begin{equation}
P_{BA} = \mathrm{Tr}(\hat{P}_k^A\tilde{\rho}_i^{(2)}) = \sum_{\ell} \mathrm{Tr}(\hat{P}_{\ell}^B\hat{P}_k^A\hat{P}_{\ell}^B\rho_i)\,.
\label{eq:PBA2}
\end{equation}
We define the corresponding order-effect asymmetry as
\begin{equation}
\Delta P_i^{(2)} = P_A - P_{BA} = \mathrm{Tr}(\hat{D}^{(2)}\,\rho_i)\,,
\label{eq:DeltaP2}
\end{equation}
where $\hat{D}^{(2)} = \hat{P}_k^A - \sum_{\ell}\hat{P}_{\ell}^B\hat{P}_k^A\hat{P}_{\ell}^B$.
This formulation corresponds more closely to the survey design of Moore~\cite{moore2002measuring}, where the response rate to $Q_A$ is compared across two independent groups: one answering $Q_A$ alone and one answering $Q_B$ before $Q_A$, without conditioning on the response to $Q_B$.
Both $\Delta P_i^{(1)}$ and $\Delta P_i^{(2)}$ vanish when $[\hat{\Pi}_A,\hat{\Pi}_B] = 0$, confirming that order effects only take place for incompatible operators.

\section{Governing equations of opinion evolution}
\label{sec:models}
\subsection{The Lindblad master equation}

We define the raising and lowering operators
\begin{equation}
\hat{\sigma}^+ = \ketbra{0}{1}\,, \quad \hat{\sigma}^- = \ketbra{1}{0}\,,
\label{eq:ladder}
\end{equation}
where $\hat{\sigma}^+$ raises from opinion~B to opinion~A, and $\hat{\sigma}^-$ lowers from opinion~A to opinion~B.
We define the two-body jump operators
\begin{equation}
\hat{L}_{ij}^+ = \hat{\sigma}_i^+\hat{\sigma}_j^-\,, \quad \hat{L}_{ij}^- = \hat{\sigma}_i^-\hat{\sigma}_j^+\,,
\label{eq:jump_social}
\end{equation}
where $\hat{\sigma}_i^\pm$ acts on $\mathcal{H}_i$ and $\hat{\sigma}_j^\mp$ acts on $\mathcal{H}_j$, with the identity on all remaining subsystems.
The operator $\hat{L}_{ij}^+$ drives the joint state from $\ket{1}_i\ket{0}_j$ to $\ket{0}_i\ket{1}_j$, exchanging the opinions of agents~$i$ and~$j$ in analogy with spin-exchange models~\cite{castellano2009statistical}. 
The operator $\hat{L}_{ij}^-$ encodes the reverse exchange.
We consider the Ising Hamiltonian on the full space $\mathcal{H}$,
\begin{equation}
\hat{H} = -\sum_{(i,j)\in E} w_{ij}\,\hat{O}_i\hat{O}_j\,,
\label{eq:hamiltonian}
\end{equation}
with energy expectation $\langle \hat{H} \rangle = -\sum_{(i,j)\in E} w_{ij}\, x_i x_j$, which is minimized when all adjacent pairs $(i,j)$ satisfy $x_ix_j=1$.
For a connected graph, the global minimum corresponds to full consensus at $x_i = 1$ or $x_i = -1$ for all~$i$.

Since social influence is irreversible and coherence is continuously lost through interaction, we model the joint dynamics with the Lindblad master equation~\cite{lindblad1976generators,gorini1976completely}.
The joint state $\rho$ evolves according to
\begin{equation}
\begin{aligned}
\dot{\rho} = -i[\hat{H},\rho] 
+ \sum_{(i,j)\in E}(1-\lambda_i)w_{ij}\Bigl(\mathcal{D}[\hat{L}_{ij}^+][\rho] + \mathcal{D}[\hat{L}_{ij}^-][\rho]\Bigr) 
+ \sum_i \lambda_i \Bigl(\gamma_{i,0}^+\mathcal{D}[\hat{\sigma}_i^+][\rho] + \gamma_{i,0}^-\mathcal{D}[\hat{\sigma}_i^-][\rho]\Bigr),
\end{aligned}
\label{eq:master_full}
\end{equation}
where $\lambda_i \in [0,1]$ is the stubbornness of agent~$i$, $\gamma_{i,0}^+ = (1+s_i)/2$ and $\gamma_{i,0}^- = (1-s_i)/2$ with $s_i \in [-1,1]$, and $\mathcal{D}[\hat{L}][\rho]$ is the Lindblad dissipator, defined as 
\begin{equation}
\mathcal{D}[\hat{L}][\rho] = \hat{L}\rho\hat{L}^\dagger - \tfrac{1}{2}\hat{L}^\dagger\hat{L}\rho - \tfrac{1}{2}\rho\hat{L}^\dagger\hat{L}\,.
\label{eq:dissipator}
\end{equation}
Equation~\eqref{eq:master_full} consists of three parts: a unitary Hamiltonian term encoding coherent social alignment, a pairwise interaction term that nudges agent~$i$ toward neighbor~$j$ via $\hat{L}_{ij}^\pm$, and an anchor term that pulls each agent toward its anchor opinion.

\subsection{Exact dynamics of joint opinion probabilities}

Let $\mathbf{b} = (b_1,\ldots,b_n) \in \{0,1\}^n$ be a binary sequence of length $n$ and let $\ket{\mathbf{b}} = \bigotimes_{i=1}^n\ket{b_i} \in \mathcal{H}$ denote the corresponding computational basis vector.
The density matrix $\rho(t)$ contains two types of information: the diagonal elements $p_\mathbf{b}(t) = \bra{\mathbf{b}}\rho(t)\ket{\mathbf{b}}$, which are directly observable probabilities, and the off-diagonal elements, which encode quantum coherence.
In addition, $\{p_\mathbf{b}(t)\}_{\mathbf{b}\in\{0,1\}^n}$ is a probability distribution over the $2^n$ configurations in the space of $\{0,1\}^n$, and the individual opinion $x_i$ is a linear combination of $p_\mathbf{b}$ via
\begin{equation}
x_i(t) = \sum_{\mathbf{b}:\,b_i=0} p_\mathbf{b}(t) - \sum_{\mathbf{b}:\,b_i=1} p_\mathbf{b}(t)\,.
\label{eq:xi_from_p}
\end{equation}

We show that the diagonal elements $p_\mathbf{b}$ satisfy a closed, exact equation of motion that is entirely decoupled from the off-diagonal elements.
In the Lindblad framework, the diagonal and off-diagonal elements of a density matrix are coupled through the Hamiltonian. 
Here, the Ising Hamiltonian~\eqref{eq:hamiltonian} is diagonal in the computational basis, which causes the coherent term to vanish from the diagonal dynamics, leaving a self-contained classical system.

\begin{prop}\label{prop:markov}
The joint probabilities $\{p_\mathbf{b}(t)\}_{\mathbf{b}\in\{0,1\}^n}$ satisfy a closed linear ODE system that is independent of all off-diagonal elements of $\rho(t)$.
Specifically, they evolve according to
\begin{equation}
\dot{p}_\mathbf{b}(t) = \sum_{\mathbf{c}\neq\mathbf{b}}
\Bigl(T_{\mathbf{b}\mathbf{c}}\,p_\mathbf{c}(t)
- T_{\mathbf{c}\mathbf{b}}\,p_\mathbf{b}(t)\Bigr)\,,
\label{eq:markov_chain}
\end{equation}
where 
\begin{equation}
T_{\mathbf{b}\mathbf{c}} = \sum_{\hat{L}\in\mathcal{L}}
\gamma_{\hat{L}}\,|\bra{\mathbf{b}}\hat{L}\ket{\mathbf{c}}|^2\,,
\label{eq:transition_rates}
\end{equation}
and $\mathcal{L}$ denotes the full set of jump operators appearing in~\eqref{eq:master_full} with their associated rates $\gamma_{\hat{L}}$.
\end{prop}

\begin{proof}
Taking $\bra{\mathbf{b}}\cdot\ket{\mathbf{b}}$ of~\eqref{eq:master_full} and using $p_\mathbf{b} = \bra{\mathbf{b}}\rho\ket{\mathbf{b}}$, we obtain
\begin{equation}
\dot{p}_\mathbf{b}
= -i\bra{\mathbf{b}}[\hat{H},\rho]\ket{\mathbf{b}}
+ \sum_{\hat{L}\in\mathcal{L}}\gamma_{\hat{L}}\,
\bra{\mathbf{b}}\mathcal{D}[\hat{L}][\rho]\ket{\mathbf{b}}\,.
\label{eq:dpdt_raw}
\end{equation}
We treat each term separately.

\textit{Hamiltonian term.}
Since $\hat{H} = -\sum_{(i,j)\in E}w_{ij}\hat{O}_i\hat{O}_j$
is diagonal in the computational basis with eigenvalues
$E_\mathbf{b} = -\sum_{(i,j)\in E}w_{ij}(1-2b_i)(1-2b_j)$,
we have $\hat{H}\ket{\mathbf{b}} = E_\mathbf{b}\ket{\mathbf{b}}$ and therefore
\begin{equation}
\bra{\mathbf{b}}[\hat{H},\rho]\ket{\mathbf{b}}
= \bra{\mathbf{b}}\hat{H}\rho\ket{\mathbf{b}}
- \bra{\mathbf{b}}\rho\hat{H}\ket{\mathbf{b}}
= E_\mathbf{b}\,p_\mathbf{b} - E_\mathbf{b}\,p_\mathbf{b} = 0\,.
\end{equation}

\textit{Dissipator term.}
Fix a jump operator $\hat{L}\in\mathcal{L}$ with rate $\gamma_{\hat{L}}$.
Each $\hat{L}$ maps every computational basis state to either zero or exactly one other state, determined by the jump operator, which means $\hat{\sigma}_i^+\ket{\mathbf{c}} = \ket{\mathbf{c}^{(i)}}$ if $c_i=1$ and $0$ otherwise; $\hat{L}_{ij}^+\ket{\mathbf{c}} = \ket{\mathbf{c}^{(ij)}}$ if $c_i=1,c_j=0$ and $0$ otherwise; and symmetrically for $\hat{\sigma}_i^-$ and $\hat{L}_{ij}^-$.
In particular, we have
\begin{equation}
\bra{\mathbf{b}}\hat{L}\ket{\mathbf{c}} \in \{0,1\}
\quad \text{for all } \mathbf{b},\mathbf{c}\in\{0,1\}^n\,,
\label{eq:binary_elements}
\end{equation}
and for each $\mathbf{b}$ there is at most one preimage $\mathbf{c}(\mathbf{b})$ with $\bra{\mathbf{b}}\hat{L}\ket{\mathbf{c}(\mathbf{b})} = 1$.

We expand $\bra{\mathbf{b}}\mathcal{D}[\hat{L}][\rho]\ket{\mathbf{b}}$ using $\hat{I} = \sum_{\mathbf{c}}\ket{\mathbf{c}}\bra{\mathbf{c}}$.
Then the first term in the dissipator yields
\begin{align}
\bra{\mathbf{b}}\hat{L}\rho\hat{L}^\dagger\ket{\mathbf{b}}
= \sum_{\mathbf{c},\mathbf{c}'}
\bra{\mathbf{b}}\hat{L}\ket{\mathbf{c}}\,
\bra{\mathbf{c}}\rho\ket{\mathbf{c}'}\,
\bra{\mathbf{c}'}\hat{L}^\dagger\ket{\mathbf{b}} \notag = \sum_{\mathbf{c},\mathbf{c}'}
\bra{\mathbf{b}}\hat{L}\ket{\mathbf{c}}\,
\overline{\bra{\mathbf{b}}\hat{L}\ket{\mathbf{c}'}}\,
\rho_{\mathbf{c}\mathbf{c}'}\,.
\label{eq:gain_expand}
\end{align}
By~\eqref{eq:binary_elements}, both
$\bra{\mathbf{b}}\hat{L}\ket{\mathbf{c}}$ and
$\overline{\bra{\mathbf{b}}\hat{L}\ket{\mathbf{c}'}}$
are nonzero only when
$\mathbf{c} = \mathbf{c}' = \mathbf{c}(\mathbf{b})$,
so the double sum collapses to a single diagonal term, which is
\begin{equation}
\bra{\mathbf{b}}\hat{L}\rho\hat{L}^\dagger\ket{\mathbf{b}}
= \rho_{\mathbf{c}(\mathbf{b})\mathbf{c}(\mathbf{b})}
= p_{\mathbf{c}(\mathbf{b})}\,.
\label{eq:gain_final}
\end{equation}
For the loss term, we have
\begin{equation}
\bra{\mathbf{b}}\hat{L}^\dagger\hat{L}\rho\ket{\mathbf{b}} = \sum_{\mathbf{c}}\bra{\mathbf{b}}\hat{L}^\dagger\hat{L}\ket{\mathbf{c}}\bra{\mathbf{c}}\rho\ket{\mathbf{b}} = \sum_{\mathbf{c}}\bra{\mathbf{b}}\hat{L}^\dagger\hat{L}\ket{\mathbf{c}} \rho_{\mathbf{c}\mathbf{b}}\,.
\end{equation}
Since $\hat{L}^\dagger\hat{L}$ is diagonal in the computational basis, $\bra{\mathbf{b}}\hat{L}^\dagger\hat{L}\ket{\mathbf{c}} = 0$ for all $\mathbf{b} \neq \mathbf{c}$.
Therefore, the anticommutator term reads
\begin{equation}
-\tfrac{1}{2}\bra{\mathbf{b}}\{\hat{L}^\dagger\hat{L},\rho\}\ket{\mathbf{b}} = -\bra{\mathbf{b}}\hat{L}^\dagger\hat{L}\ket{\mathbf{b}}\,p_\mathbf{b}\,.
\label{eq:loss_final}
\end{equation}
We combine~\eqref{eq:gain_final} and~\eqref{eq:loss_final}. It follows that $\bra{\mathbf{b}}\mathcal{D}[\hat{L}][\rho]\ket{\mathbf{b}}$ depends only on the diagonal elements of $\rho$.
We sum over all $\hat{L}\in\mathcal{L}$ and substitute into~\eqref{eq:dpdt_raw} to obtain
\begin{equation}
\dot{p}_\mathbf{b} =  \sum_{\hat{L}\in\mathcal{L}}\gamma_{\hat{L}}
\Bigl(p_{\mathbf{c}(\mathbf{b})} - \bra{\mathbf{b}}\hat{L}^\dagger\hat{L}\ket{\mathbf{b}}\,p_\mathbf{b}\Bigr)\,.
\label{eq:dpdt_final}
\end{equation}
The gain term $\sum_{\hat{L}}\gamma_{\hat{L}}\,p_{\mathbf{c}(\mathbf{b})}$ counts all jump operators that map some configuration $\mathbf{c}$ to $\mathbf{b}$, weighted by their rates, which is precisely $\sum_{\mathbf{c}\neq\mathbf{b}}T_{\mathbf{bc}}\,p_\mathbf{c}$ with $T_{\mathbf{bc}}$ as defined in~\eqref{eq:transition_rates}.
For the loss term, we use the identity
$\bra{\mathbf{b}}\hat{L}^\dagger\hat{L}\ket{\mathbf{b}} = \|\hat{L}\ket{\mathbf{b}}\|^2$.
By the structure of the jump operators, $\hat{L}\ket{\mathbf{b}}$ is either $0$ or a single computational basis state $\ket{\mathbf{c}(\mathbf{b})}$.
If $\hat{L}\ket{\mathbf{b}} = 0$, then $\hat{L}$ contributes nothing to the loss.
If $\hat{L}\ket{\mathbf{b}} = \ket{\mathbf{c}(\mathbf{b})}\neq 0$, then by orthonormality of the computational basis, $\bra{\mathbf{c}}\hat{L}\ket{\mathbf{b}} = 1$ only if $\mathbf{c}=\mathbf{c}(\mathbf{b})$.
Therefore $\bra{\mathbf{b}}\hat{L}^\dagger\hat{L}\ket{\mathbf{b}} \in \{0,1\}$, and summing over all $\hat{L}\in\mathcal{L}$ gives
\begin{equation}
\sum_{\hat{L}\in\mathcal{L}}\gamma_{\hat{L}}\bra{\mathbf{b}}\hat{L}^\dagger\hat{L}\ket{\mathbf{b}}
= \sum_{\hat{L}:\,\hat{L}\ket{\mathbf{b}}\neq 0}\gamma_{\hat{L}}
= \sum_{\mathbf{c}\neq\mathbf{b}}T_{\mathbf{cb}}\,.
\label{eq:loss_sum}
\end{equation}
We combine~\eqref{eq:loss_final} and~\eqref{eq:loss_sum} and obtain $\sum_{\mathbf{c}\neq\mathbf{b}}T_{\mathbf{cb}}\,p_\mathbf{b}$ for the loss term, completing the identification with~\eqref{eq:markov_chain}.
\end{proof}

\subsection{Markov chain and opinion correlations}

Note that Eq.~\eqref{eq:markov_chain} defines a continuous-time Markov chain on $\{0,1\}^n$ with transition rates $T_{\mathbf{bc}}$ given in~\eqref{eq:transition_rates}.
The off-diagonal entries of the generator $M \in \mathbb{R}^{2^n\times 2^n}$ are $M_{\mathbf{bc}} = T_{\mathbf{bc}} \geq 0$ for $\mathbf{b}\neq\mathbf{c}$, which follows directly from~\eqref{eq:transition_rates}.
The diagonal entries are $M_{\mathbf{bb}} = -\sum_{\mathbf{c}\neq\mathbf{b}}T_{\mathbf{cb}}$, so each column sums to zero by construction.
Therefore $M$ is a valid Markov generator, and $\{p_\mathbf{b}(t)\}$ remains a probability distribution for all $t \geq 0$.

The nonzero transition rates arise from two sources.
For each agent~$i$, the operators $\hat{\sigma}_i^\pm$ induce single-bit flips between $\mathbf{b}$ and $\mathbf{b}^{(i)}$, where $\mathbf{b}^{(i)}$ denotes $\mathbf{b}$ with bit $i$ flipped, yielding anchor transitions
\begin{equation}
T_{\mathbf{b}^{(i)},\,\mathbf{b}} = \begin{cases} \lambda_i\,\gamma_{i,0}^+ & \text{if } b_i = 1 \\ \lambda_i\,\gamma_{i,0}^- & \text{if } b_i = 0 \end{cases}
\label{eq:anchor_transitions}
\end{equation}
with $\gamma_{i,0}^\pm = (1\pm s_i)/2$.
For each directed edge $(i,j)\in E$, the operator $\hat{L}_{ij}^+$ induces a simultaneous two-bit flip, where agent~$i$ moves from opinion~B to opinion~A while agent~$j$ moves from opinion~A to opinion~B, yielding social influence transitions
\begin{equation}
T_{\mathbf{b}^{(ij)},\,\mathbf{b}} = (1-\lambda_i)\,w_{ij}\cdot\mathbf{1}[b_i=1,\,b_j=0]\,,
\label{eq:social_transitions}
\end{equation}
and symmetrically for $\hat{L}_{ij}^-$.
For $\lambda_i > 0$ and $s_i \in (-1,+1)$ for all~$i$, the anchor rates are strictly positive, making the chain irreducible on $\{0,1\}^n$ and guaranteeing convergence to a unique stationary distribution $\mathbf{p}_\infty$ satisfying $M\mathbf{p}_\infty = 0$.

Since $\{p_\mathbf{b}(t)\}$ is a probability distribution over $\{0,1\}^n$, all joint moments of agent opinions are well defined.
We identify the opinion of agent~$i$ as the random variable $X_i = 1 - 2b_i \in \{-1,+1\}$, so that the expected opinion is
\begin{equation}
x_i(t) = \mathbb{E}[X_i] = \sum_{\mathbf{b}\in\{0,1\}^n}(1-2b_i)\,p_\mathbf{b}(t) = \mathrm{Tr}(\rho\,\hat{O}_i)\,.
\label{eq:xi_from_markov}
\end{equation}
We define the opinion covariance between agents~$i$ and~$j$ as
\begin{equation}
\mathrm{Cov}(X_i, X_j) = \mathbb{E}[X_i X_j] - \mathbb{E}[X_i]\mathbb{E}[X_j] = \mathrm{Tr}(\rho\,\hat{O}_i\hat{O}_j) - x_i x_j\,,
\label{eq:covariance}
\end{equation}
where the joint expectation $\mathbb{E}[X_i X_j] = \sum_{\mathbf{b}}(1-2b_i)(1-2b_j)\,p_\mathbf{b} = \mathrm{Tr}(\rho\,\hat{O}_i\hat{O}_j)$ is a two-body quantum observable accessible from the full joint state $\rho$.
The covariance~\eqref{eq:covariance} measures the statistical dependence between agents' opinion configurations beyond what individual opinions $x_i$ capture, and $\mathrm{Cov}(X_i,X_j)$ vanishes under the product state approximation~\eqref{eq:product_ansatz}, since $\mathrm{Tr}(\bigotimes_k\rho_k\,\hat{O}_i\hat{O}_j) = x_i x_j$. The two-body social transitions~\eqref{eq:social_transitions} generate nonzero covariance by coupling connected agents' opinion configurations, building correlations that persist at steady state.

\subsection{Reduced dynamics under the product state approximation}\label{sec:product}

For a system of $n$ agents, the joint state $\rho$ lives in a Hilbert space of dimension $2^n$, and the master equation~\eqref{eq:master_full} governs the evolution of a density matrix with $4^n$ real parameters.
This exponential scaling makes the full joint dynamics intractable for large $n$, motivating the product state approximation
\begin{equation}
\rho \approx \bigotimes_{k=1}^n \rho_k\,,
\label{eq:product_ansatz}
\end{equation}
which assumes that the joint state factorizes into independent single-agent states at all times.
Under~\eqref{eq:product_ansatz}, the two-body term $\mathrm{Tr}_{\neq i,j}(\rho) \approx \rho_i \otimes \rho_j$ neglects inter-agent correlations, reducing the problem to $n$ coupled equations for $\rho_i \in \mathcal{D}(\mathcal{H}_i)$.

We take the partial trace $\mathrm{Tr}_{\neq i}$ of~\eqref{eq:master_full} under~\eqref{eq:product_ansatz}. The resulting reduced master equation for agent~$i$ (derived in Appendix~\ref{app:marginal}) is
\begin{equation}
\begin{aligned}
\dot{\rho}_i = &-i\bigl[\hat{H}_i^{\mathrm{eff}},\rho_i\bigr] + \sum_{j}(1-\lambda_i)w_{ij}\Bigl(\tfrac{1+x_j}{2}\,\mathcal{D}[\hat{\sigma}_i^+][\rho_i] + \tfrac{1-x_j}{2}\,\mathcal{D}[\hat{\sigma}_i^-][\rho_i]\Bigr) \\
&+ \lambda_i\Bigl(\gamma_{i,0}^+\,\mathcal{D}[\hat{\sigma}_i^+][\rho_i] + \gamma_{i,0}^-\,\mathcal{D}[\hat{\sigma}_i^-][\rho_i]\Bigr)\,,
\end{aligned}
\label{eq:master_marginal}
\end{equation}
where the effective single-site Hamiltonian
\begin{equation}
\hat{H}_i^{\mathrm{eff}} = -\Omega_i\,\hat{O}_i\,, \quad \Omega_i = \sum_j w_{ij}\,x_j\,,
\label{eq:Heff}
\end{equation}
arises from the partial trace of~\eqref{eq:hamiltonian}, with the two-body coupling $\hat{O}_i\hat{O}_j$ reducing to a local field $-\Omega_i\hat{O}_i$ weighted by the influence-weighted mean of neighbors' expressed opinions.
The coefficients $(1+x_j)/2$ and $(1-x_j)/2$ are the Born-rule probabilities that agent~$j$ holds opinion~A and opinion~B respectively, since $\mathrm{Tr}(\rho_j\hat{\sigma}_j^+\hat{\sigma}_j^-) = (1+x_j)/2$ and $\mathrm{Tr}(\rho_j\hat{\sigma}_j^-\hat{\sigma}_j^+) = (1-x_j)/2$.
These coefficients encode the social influence of neighbor~$j$ on agent~$i$ entirely through $x_j$, the expressed opinion of~$j$.

We apply $\mathrm{Tr}(\cdot\,\hat{O})$ to~\eqref{eq:master_marginal} and obtain the governing equations for opinions
\begin{align}\label{eq:xi} 
\dot{x}_i = \lambda_i(s_i - x_i) + (1-\lambda_i)\sum_j w_{ij}(x_j - x_i)\,, 
\end{align}
which is exactly the continuous-time FJ model~\cite{friedkin1990social}.

We note that~\eqref{eq:product_ansatz} is not a separable state approximation in the quantum mechanical sense, which allows convex mixtures of product states $\rho = \sum_k p_k \bigotimes_i \rho_i^{(k)}$.
Since the master equation~\eqref{eq:master_full} is linear in $\rho$, the reduced dynamics~\eqref{eq:master_marginal} extend immediately to separable states by linearity. The governing equations of opinions~\eqref{eq:xi} hold for each pure product component $\bigotimes_i \rho_i^{(k)}$, and the full separable state evolves as the same convex mixture of the resulting trajectories.

\section{Numerical experiments}
\label{sec:numerics}

\subsection{Exact dynamics on small networks}

Tracking the full joint density matrix $\rho(t) \in \mathcal{D}(\mathcal{H})$ with $\dim\mathcal{H} = 2^n$ is infeasible for large $n$.
We simulate the full Lindblad equation~\eqref{eq:master_full} on small networks with $n = 6$ agents on four undirected networks: a complete graph, a chain graph, a ring graph, and a Barbell graph (two cliques of three agents connected by a single bridge), where the graphs are illustrated in Fig.~\ref{fig:networks}.
The influence matrix $W$ is row-normalized in each case.
\begin{figure}[htp]
\centering
\includegraphics[width=0.75\linewidth]{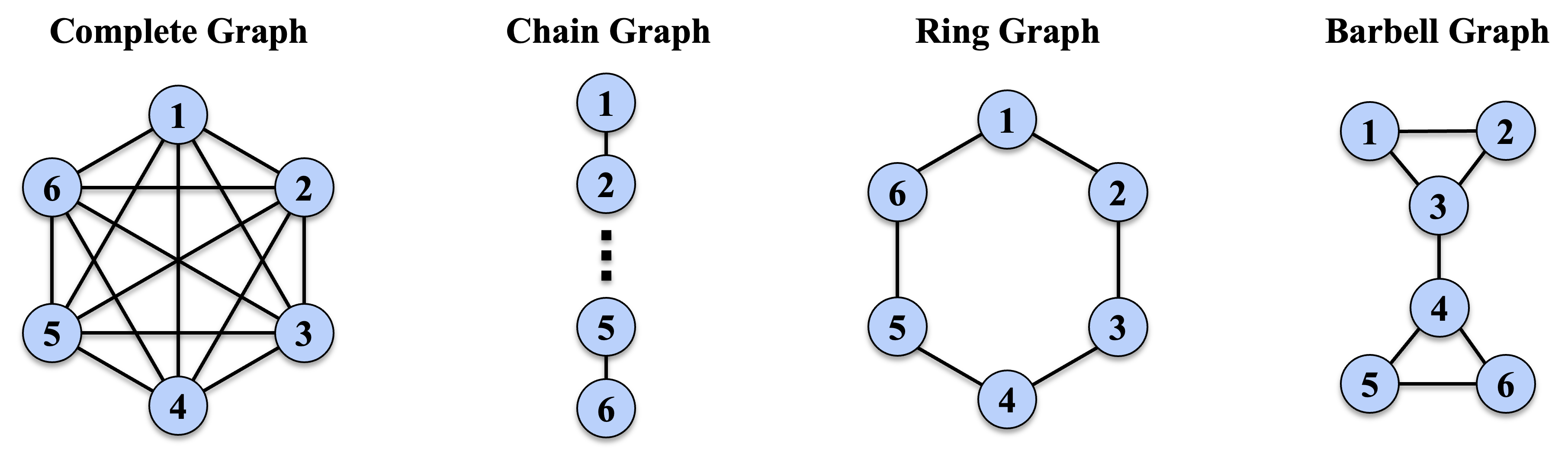}
\caption{The diagram of the four networks with 6 nodes and their indices.}
\label{fig:networks}
\end{figure}

We choose stubbornness $\lambda_i = 0.1$ for all agents and anchor opinions $s_i=0.4i-1.4$, which is uniformly spaced in $[-1, 1]$.
The initial state is a product state $\rho(0) = \bigotimes_i \rho_i(0)$, with each $\rho_i(0)$ constructed from the Bloch vector with $x_i(0) = s_i$, $c_i(0) = \sqrt{1 - s_i^2}$, and $\phi_i(0) \sim \mathcal{U}(0, 2\pi)$.
We use a fourth-order Runge--Kutta method for time integration with step size $\Delta t = 10^{-3}$.

\subsubsection{Comparison of quantum and classical models}
We simulate the exact Lindblad equation~\eqref{eq:master_full}, extract $x_i(t) = \mathrm{Tr}(\rho_i(t)\hat{O})$ from the marginal density matrices, and compare the resulting trajectories to the classical FJ model~\eqref{eq:xi} in Fig.~\ref{fig:opinions}. 
The two models follow similar trends and converge at comparable rates, but do not agree exactly, with discrepancies that depend on the network.
In the quantum model, the two-body jump operators $\hat{L}_{ij}^\pm$ generate inter-agent correlations that feed back into the marginal dynamics, causing the exact trajectories to deviate from the classical FJ model.
At steady state, the exact quantum model exhibits smaller opinion variance across agents than the classical model on all four networks.
Inter-agent correlations generated by the social jump operators act as an additional consensus-promoting mechanism, pulling opinions closer together than the product state approximation predicts.
\begin{figure}[htp]
\centering
\includegraphics[trim={20 10 90 10},clip,width=\linewidth]{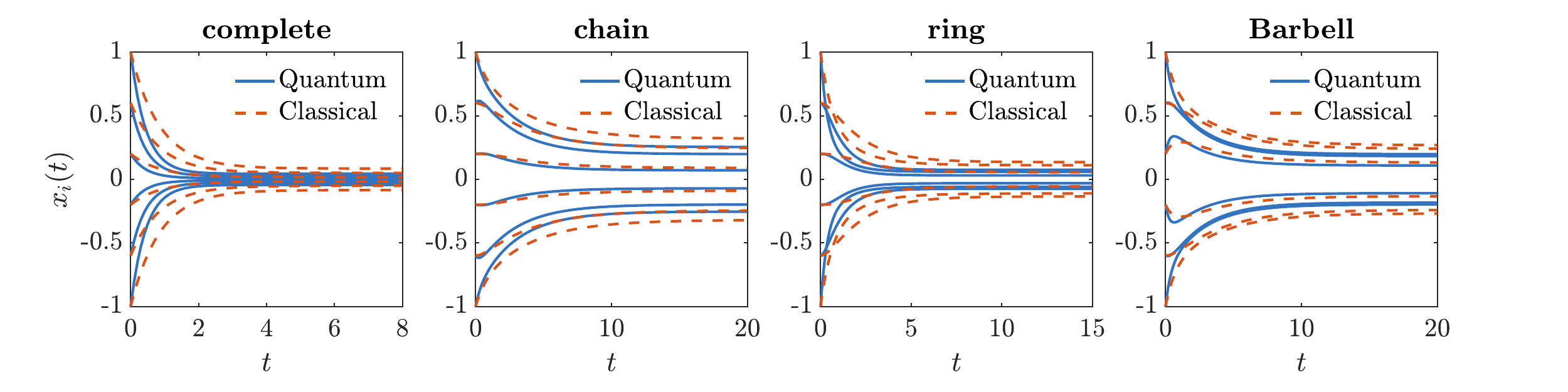}
\caption{Opinion trajectories $x_i(t)$ from the quantum opinion model \eqref{eq:master_full} and the classical FJ model \eqref{eq:xi} on different networks. All simulations have the same initial states so that the initial opinions are uniformly spaced on $[-1,1]$.}
\label{fig:opinions}
\end{figure}

We further examine the quantum model by tracking three quantities that have no direct classical counterparts (except for the energy).
The first is the Frobenius distance $\|\rho(t) - \rho_{\mathrm{prod}}(t)\|_F$, which measures the total deviation of the joint state from the product state approximation, where
\begin{equation}
\rho_{\mathrm{prod}}(t) = \bigotimes_i \rho_i(t)\,, \quad \rho_i(t) = \mathrm{Tr}_{\neq i}(\rho(t))\,.
\end{equation}
The second is the energy $\langle\hat{H}\rangle(t) = -\sum_{(i,j)\in E} w_{ij} x_i(t) x_j(t)$, which measures the degree of opinion alignment within the network.
The third is the quantum mutual information between two agent groups $C_1 = \{1,2,3\}$ and $C_2 = \{4,5,6\}$,
\begin{equation}
I(C_1 : C_2)(t) = S(\rho_{C_1}(t)) + S(\rho_{C_2}(t)) - S(\rho(t))\,,
\label{eq:mutual_info}
\end{equation}
where $S(\rho) = -\mathrm{Tr}(\rho\log\rho)$ is the von Neumann entropy and $\rho_{C_k}(t) = \mathrm{Tr}_{\neq C_k}(\rho(t))$ is the marginal state of group $C_k$.
The quantum mutual information quantifies the total correlations between the two groups~\cite{nielsen2010quantum}, capturing both classical statistical dependence and quantum coherences shared across the partition.
Fig.~\ref{fig:combined} shows their time trajectories.
\begin{figure}[htp]
\centering
\includegraphics[trim={0 0 50 0},clip,width=0.95\linewidth]{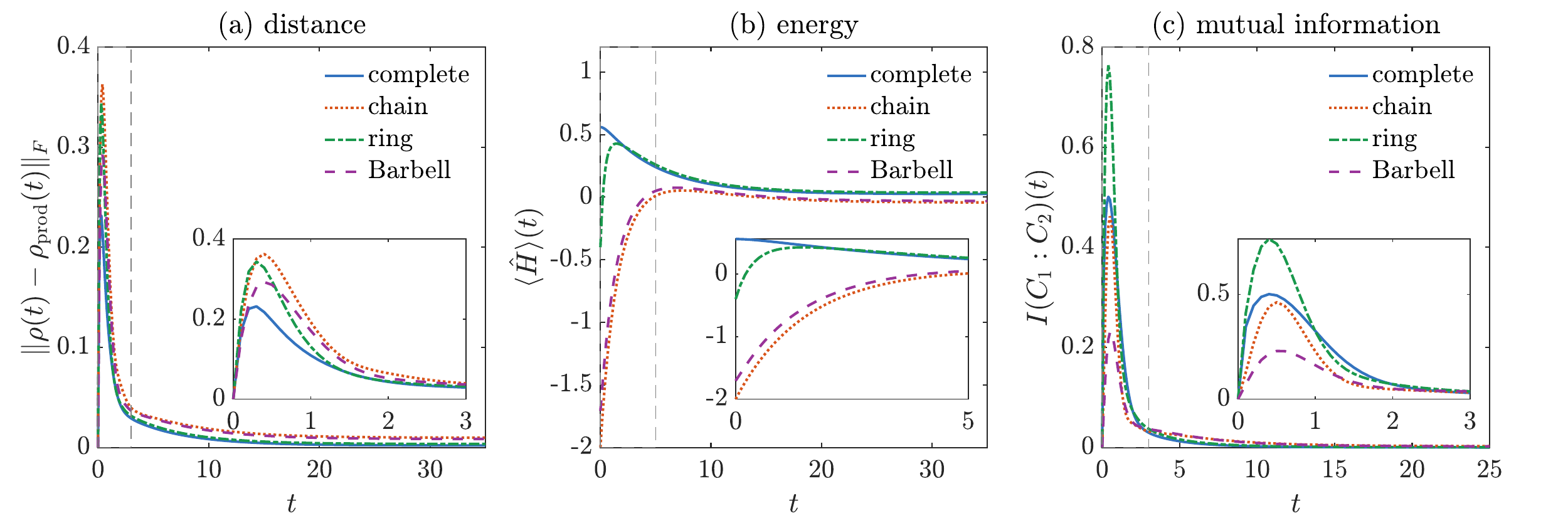}
\caption{Frobenius distance, quantum mutual information $I(C_1:C_2)(t)$, and energy $\langle\hat{H}\rangle(t)$ over time for each network.}
\label{fig:combined}
\end{figure}

The Frobenius distance starts at zero due to the product state initial condition and grows as the social jump operators generate inter-agent correlations.
It saturates at a network-dependent nonzero steady-state value, indicating that the two-body jump operators $\hat{L}_{ij}^{\pm}$ generate persistent correlations not captured by the product state approximation.
The energy~\eqref{eq:hamiltonian} is minimized at full consensus $x_i = \pm 1$ for all $i$.
For the ring and Barbell networks, the energy increases visibly before settling toward its steady state.
Initially, the anchor opinions are spread across $[-1,1]$, so only the neighboring pair $(3,4)$ contributes a positive term $-w_{34}x_3x_4$ to the energy while all other pairs contribute negative terms, resulting in a relatively small initial energy.
As opinions evolve toward the center, the products $x_ix_j$ decrease in magnitude, raising the energy before social influence eventually drives agents toward alignment.
The quantum mutual information grows initially as the two groups become correlated through social interactions and saturates at a network-dependent value.
The Barbell network yields the least inter-group correlation across all four networks.
Since $C_1$ and $C_2$ correspond to the two cliques connected by a single bridge edge, inter-group information flow is severely limited, suppressing the build-up of correlations across the partition.
For the complete and ring networks, the partition $C_1 = \{1,2,3\}$, $C_2 = \{4,5,6\}$ imposes no community structure, so correlations build up strongly across the partition.

\subsubsection{Pairwise opinion covariance}
We simulate the exact Markov chain~\eqref{eq:markov_chain} governing the joint opinion probabilities $\{p_\mathbf{b}(t)\}$ and compute the opinion covariance $\mathrm{Cov}(X_i, X_j)$ defined in~\eqref{eq:covariance}.
We initialize the system in a single-agent contagion scenario with $p_\mathbf{b}(0) = \delta_{\mathbf{b}, \mathbf{e}_1}$, where $\mathbf{e}_1 = (1,0,\ldots,0)$ denotes the configuration in which only agent~$1$ holds opinion~B, corresponding to $x_1(0) = -1$ and $x_i(0) = 1$ for $i \geq 2$.
We set $\lambda_i = 0$ for all agents and plot the time-dependent covariance between agent~$1$ and all other agents in Fig.~\ref{fig:covariance}.
Agents with symmetric roles in the network have identical covariance curves.
For agents directly connected to agent~$1$, such as agent~$2$ in the chain graph and agents~$2$ and~$6$ in the ring graph, the covariance drops rapidly at early times.
This occurs because the social jump operators $\hat{L}_{1j}^\pm$ transfer opinion~B directly from agent~$1$ to its neighbors, creating configurations in which $X_1 = 1$ and $X_j = 1$, which drives $\mathbb{E}[X_1 X_j]$ negative immediately.
Agents not directly connected to agent~$1$ receive opinion~B via multi-node paths, so their covariance with agent~$1$ decreases more slowly and peaks later in time.
\begin{figure}[htp]
\centering
\includegraphics[trim={0 10 30 10},clip,width=\linewidth]{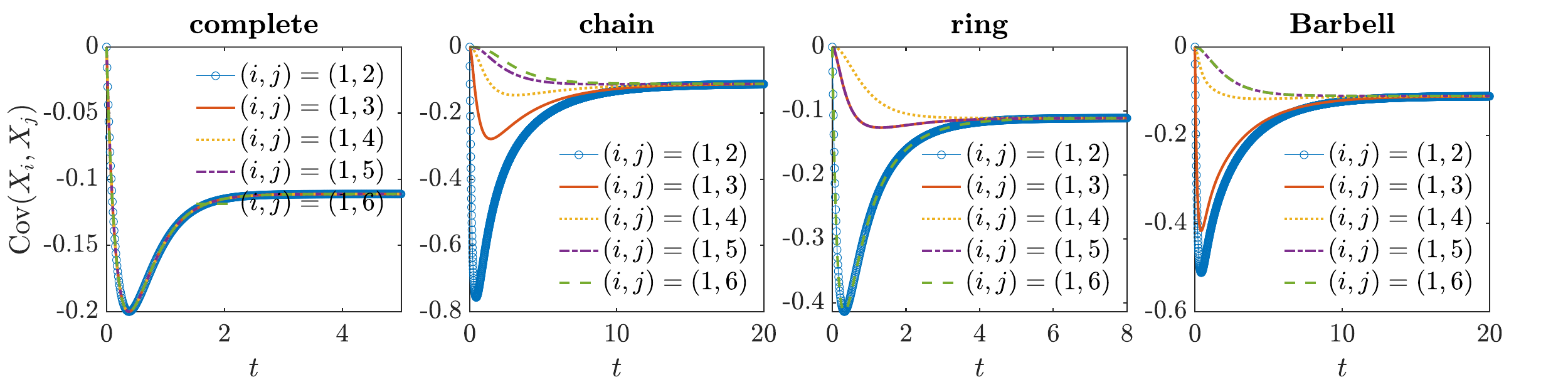}
\caption{Pairwise opinion covariance $\mathrm{Cov}(X_i, X_j)(t)$ between agent 1 and other agents. }
\label{fig:covariance}
\end{figure}

In addition, the covariance trajectories differ across networks during the transient but converge to the same steady-state value.
This is a direct consequence of the conservation of total opinion mass $\sum_i b_i$ when $\lambda_i = 0$: the dynamics are confined to the subspace $\mathcal{S}_1 = \{\mathbf{b} \in \{0,1\}^n : \sum_i b_i = 1\}$, within which the stationary distribution is uniform over all $n$ configurations by the exchange symmetry of $\hat{L}_{ij}^\pm$ on undirected graphs. The steady-state covariance is therefore network-independent.
For the uniform distribution over $\mathcal{S}_1$, $p_\infty(\mathbf{e}_k) = 1/n$ for all $k$, direct computations yield
\begin{equation}
\begin{aligned}
\mathbb{E}_\infty[X_i] &= \sum_{k=1}^n \frac{1}{n}(1 - 2\delta_{ik}) = \frac{n-2}{n}\,, \\
\mathbb{E}_\infty[X_i X_j] &= \sum_{k=1}^n \frac{1}{n}(1-2\delta_{ik})(1-2\delta_{jk}) = \frac{n-4}{n}\,,
\label{eq:cov_steady}
\end{aligned}
\end{equation}
for $i \neq j$ and therefore 
\begin{equation}
\mathrm{Cov}_\infty(X_i, X_j) = \frac{n-4}{n} - \left(\frac{n-2}{n}\right)^2 = -\frac{4}{n^2}\,,
\end{equation}
which evaluates to $-1/9 \approx -0.111$ for $n = 6$, consistent with the numerical results in Fig.~\ref{fig:covariance}.

\subsection{Reduced dynamics on a Facebook network}

We simulate the quantum opinion model under the product state approximation~\eqref{eq:master_marginal} on a large real-world network.
We use the largest connected component of the Caltech network from the Facebook-100 dataset~\cite{traud2012social,red2011comparing}, where nodes represent individuals and edges encode Facebook friendships recorded on a single day in the fall of 2005, giving $n = 769$ agents.
The influence matrix $W$ is obtained by row-normalizing the binary adjacency matrix.

We partition agents into $K$ communities via spectral clustering on the normalized Laplacian and assign anchor opinions community-wise. 
Community means are $\mu_k = -1 + 2k/(K+1)$ for $k = 1,\ldots,K$, and the anchor of each agent in community $k$ is drawn as $s_i \sim \mathcal{N}(\mu_k, 0.16)$, clipped to $[-1,+1]$.
We initialize opinions at the anchor values, $x_i(0) = s_i$, following the original FJ formulation~\cite{friedkin1990social}.
Stubbornness parameters are drawn independently as $\lambda_i \sim \mathcal{U}(0.1, 0.4)$.
For the quantum initial conditions, we set $c_i(0) = 1$ for all agents, corresponding to maximally coherent states that represent maximum cognitive ambivalence prior to any social interaction, and draw initial phases uniformly as $\phi_i(0) \sim \mathcal{U}(0, 2\pi)$.

\subsubsection{Opinion and coherence dynamics}
Fig.~\ref{fig:opinion_neighbors} shows opinion trajectories of the quantum model under the product state approximation~\eqref{eq:master_marginal}, with agents colored by community, alongside scatter plots of $x_i(t)$ against the weighted mean of their neighbors $\bar{x}_i(t) = \sum_j w_{ij} x_j(t)$ at three snapshots.
At $t=0$, the scatter plot of $x_i$ against $\bar{x}_i$ is widely dispersed, reflecting the random initialization of opinions within each community.
As the dynamics evolve, opinions within the same community align progressively, and by $t=2$ the points begin to cluster along the diagonal with visible community separation.
By $t=12$, the scatter collapses tightly onto the diagonal, indicating that each agent's opinion has converged to a value close to its neighbors' mean.
The tight alignment of $x_i$ and $\bar{x}_i$ at steady state reflects the balance between each agent's stubbornness toward its anchor opinion and the social pressure from its neighbors.

The quantum coherence $\rho_i^{01}(t)$ defined in~\eqref{eq:polar} encodes information about each agent's cognitive ambivalence that has no classical counterpart.
Under the product state approximation \eqref{eq:product_ansatz}, the coherence magnitude decays as $c_i(t) = c_i(0)e^{-t/2}$, independent of the network structure and of all other agents, as derived in Appendix~\ref{app:coherence}.
The coherence phase evolves as $\dot{\phi}_i = -2\bar{x}_i(t)$, where $\bar{x}_i(t) = \sum_j w_{ij} x_j(t)$ is the influence-weighted mean opinion of agent~$i$'s neighbors, so the rotation rate is set by the local social field.
\begin{figure}[htbp]
    \centering
    \includegraphics[width=0.19\linewidth]{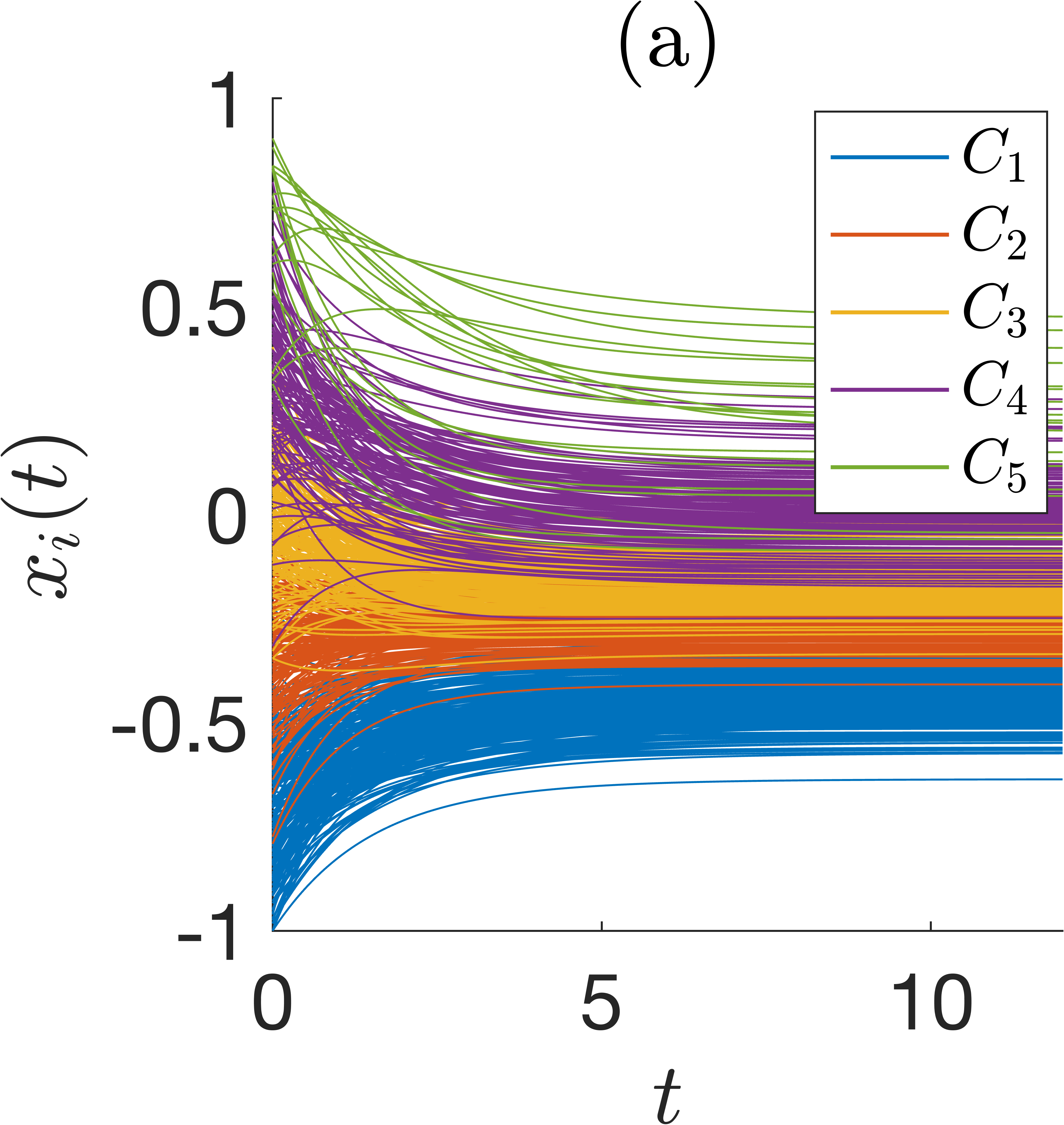}
    \includegraphics[width=0.19\linewidth]{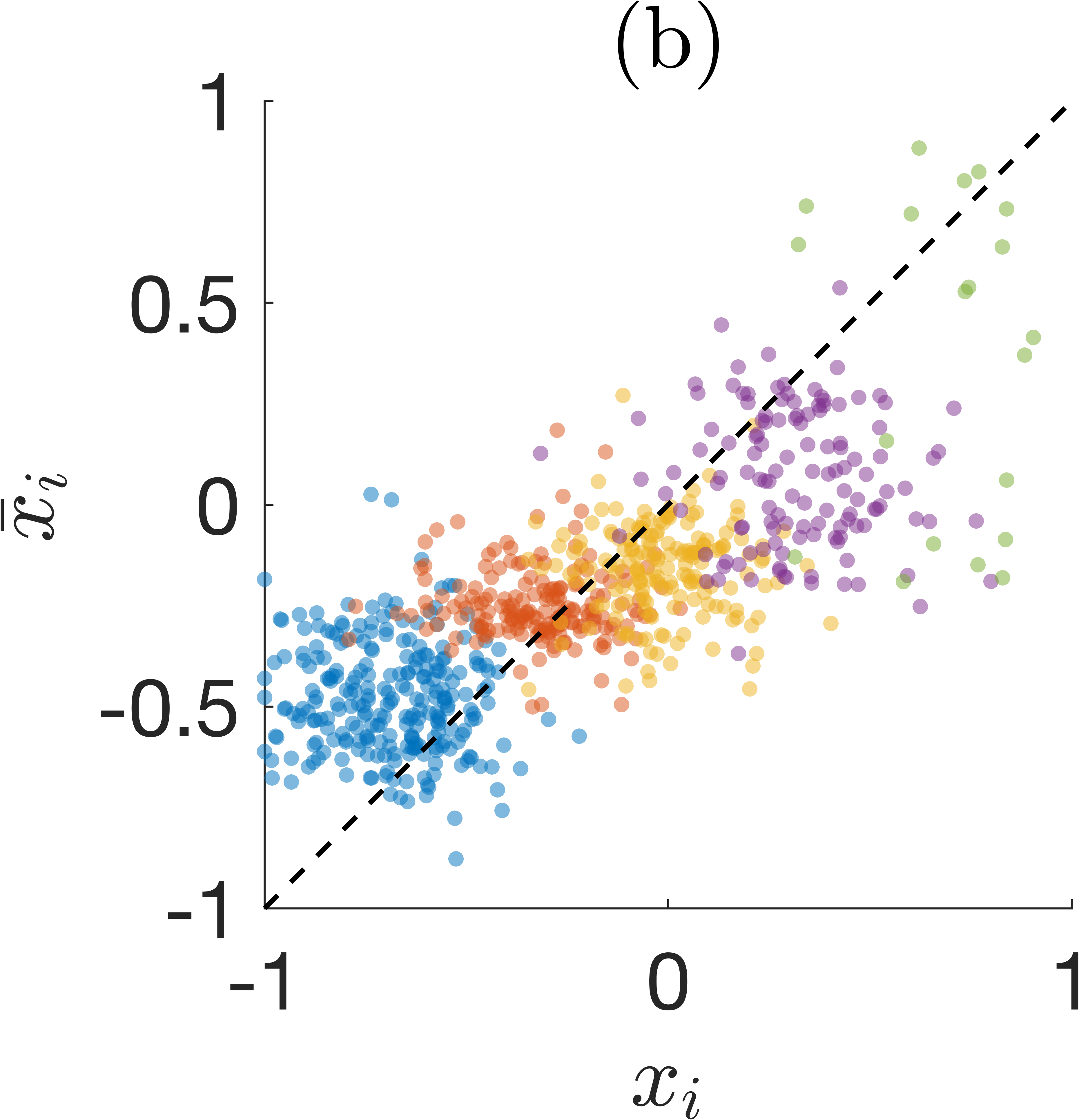}
    \includegraphics[width=0.19\linewidth]{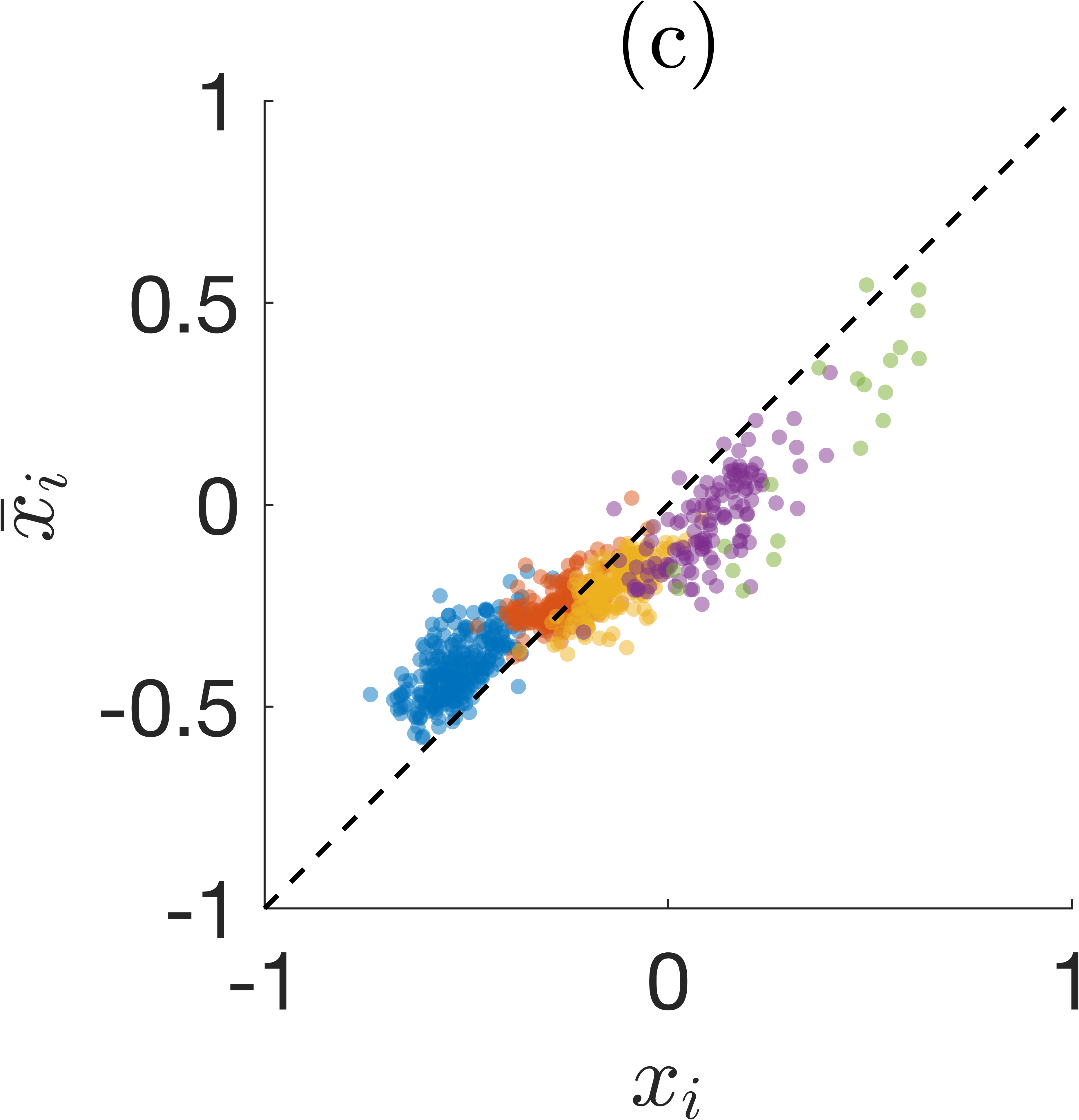}
    \includegraphics[width=0.19\linewidth]{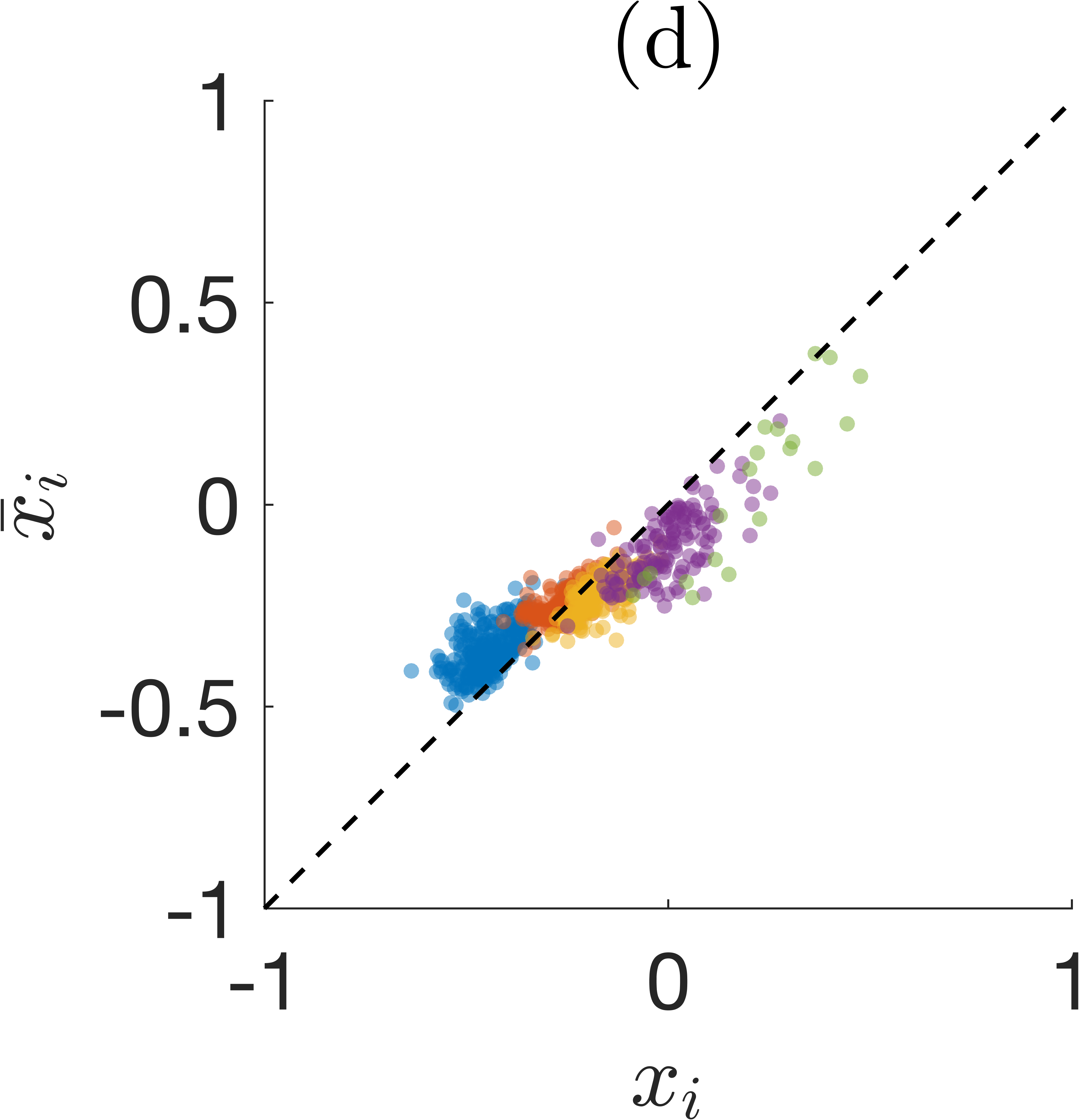}
    \includegraphics[width=0.19\linewidth]{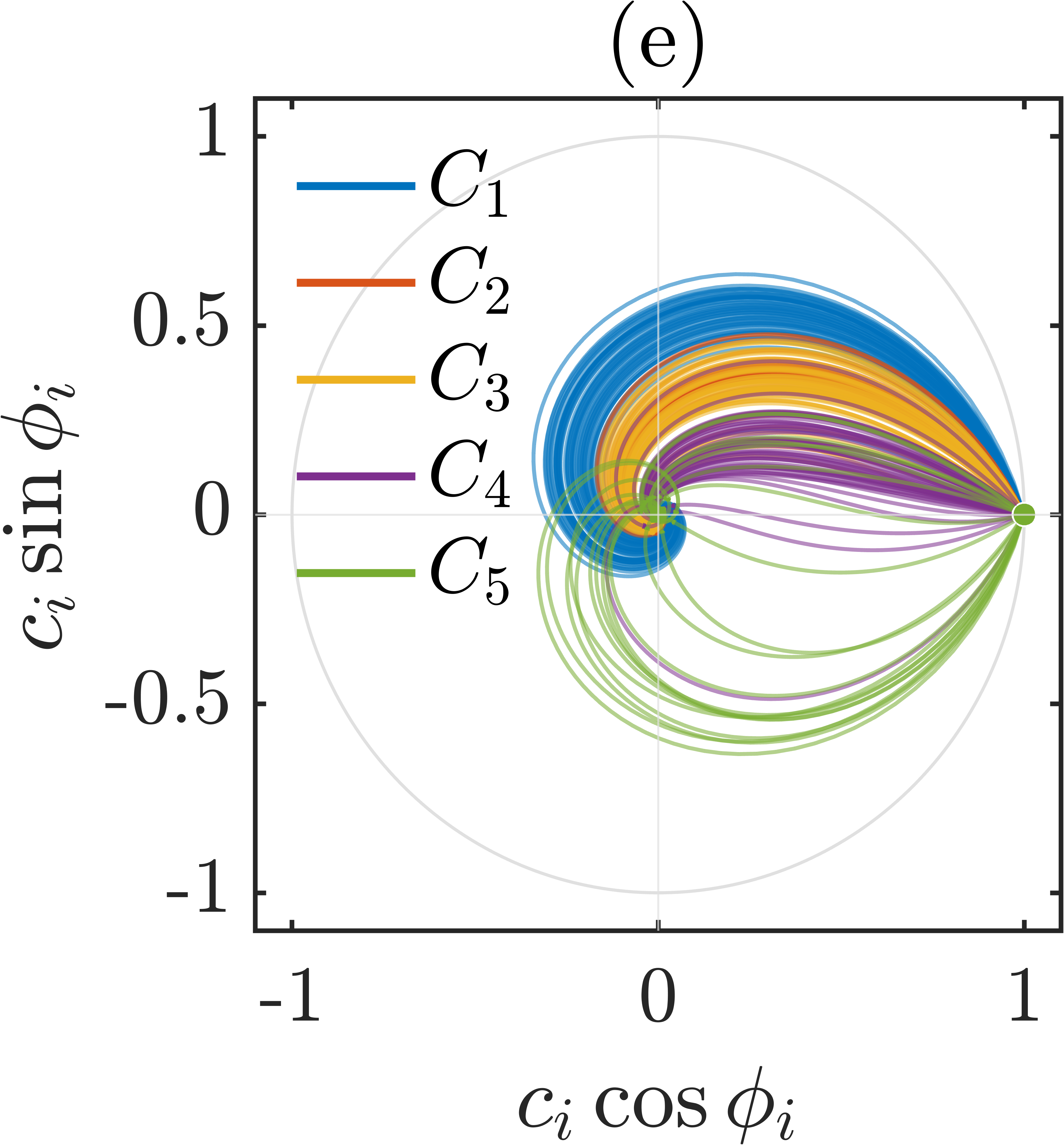}
    \caption{(a): opinion trajectories $x_i(t)$ colored by community. (b--d): scatter plots of agent opinion $x_i$ against neighbor mean opinion $\bar{x}_i$ at $t = 0$, $2$, and $12$, showing progressive alignment within communities and convergence to community-dependent steady states. (e): Trajectories of the quantum coherence $\rho_i^{01}(t)$ in the complex plane for 50 randomly selected agents per community, colored by community membership.}
    \label{fig:opinion_neighbors}
\end{figure}

Fig.~\ref{fig:opinion_neighbors}(e) shows the trajectories of $\rho_i^{01}(t) = \frac{c_i}{2}e^{-i\phi_i}$ in the complex plane for 50 randomly selected agents from each community.
Since $\phi_i(0)$ is initialized uniformly on $[0, \pi)$, we shift each trajectory so that it begins at $(1, 0)$ at $t=0$, marked by a filled circle.
As $t \to \infty$, $c_i \to 0$ and each trajectory spirals inward toward the origin, reflecting complete decoherence of the agent's cognitive state.
The rotation direction encodes the sign of the local social field.
Agents in communities with an overall positive opinion mean, such as $C_5$, rotate clockwise, while agents in communities with a negative opinion mean, such as $C_1$, rotate counterclockwise.
Agents in community $C_4$, whose opinion mean is near zero, exhibit slow and irregular rotation as the social field is weak and changes sign during the transient.
The rotation rate is not uniform over time: it is fastest during the early transient when opinion differences across communities are large, and slows as opinions converge toward their community-dependent steady states and the social field $\bar{x}_i$ stabilizes.
Communities with greater opinion extremity, such as $C_1$ and $C_5$, produce tighter and more regular spiral arms, while communities with opinion means near zero produce broader and more dispersed trajectories.

\subsubsection{Order effects}
We consider two incompatible measurement operators: the direct opinion poll $\hat{\Pi}_A = \ketbra{0}{0}$ and a differently framed question operator $\hat{\Pi}_B(\theta)$, defined as
\begin{equation}
\hat{\Pi}_B(\theta) = \ketbra{\psi(\theta)}{\psi(\theta)}\,, \quad
\ket{\psi(\theta)} = \cos\tfrac{\theta}{2}\ket{0} + \sin\tfrac{\theta}{2}\ket{1}\,,
\label{eq:PiB}
\end{equation}
where $\theta \in (0,\pi)$ is the incompatibility angle between the two questions.
Both $\hat{\Pi}_A$ and $\hat{\Pi}_B(\theta)$ are rank-1 projectors with $[\hat{\Pi}_A, \hat{\Pi}_B(\theta)] \neq 0$ for $\theta \neq 0$, so the order of measurement affects subsequent measurement results.
The operator $\hat{\Pi}_B(\theta)$ represents a similar question to $Q_A$ but posed under a slightly different cognitive framing that rotates the measurement axis by $\theta$.
At $\theta = 0$, the two questions are the same and no order effect takes place, while at $\theta = \pi/2$ they are maximally incompatible.

We substitute~\eqref{eq:PiB} into the order-effect asymmetry formulas~\eqref{eq:DeltaP1} and~\eqref{eq:DeltaP2} and obtain
\begin{equation}
\begin{aligned}
\Delta P_i^{(1)}(t;\theta) &= \tfrac{\sin^2\theta}{4}\,x_i(t) - \sin\tfrac{\theta}{2}\cos^3\tfrac{\theta}{2}\, c_i(t)\cos\phi_i(t)\,, \\
\Delta P_i^{(2)}(t;\theta) &= \tfrac{\sin^2\theta}{2}\,x_i(t) - \tfrac{\sin 2\theta}{4}\, c_i(t)\cos\phi_i(t)\,,
\end{aligned}
\label{eq:DeltaP_formula}
\end{equation}
where $x_i(t)$ is the opinion, $c_i(t)$ is the coherence magnitude, and $\phi_i(t)$ is the coherence phase.
Both asymmetries vanish when $\theta = 0$, consistent with the commutativity of $\hat{\Pi}_A$ and $\hat{\Pi}_B(0)$.
Each asymmetry has two contributions with distinct interpretations.
The first term, proportional to $x_i(t)$, reflects the agent's opinion and persists even in the absence of quantum coherence.
The second term, proportional to $c_i(t)\cos\phi_i(t)$, is a purely quantum contribution that vanishes when $c_i = 0$ and depends on the coherence phase $\phi_i$, encoding the direction of cognitive lean on the equatorial plane of the Bloch sphere.
The two contributions can reinforce or cancel depending on the sign of $\cos\phi_i(t)$ relative to $x_i(t)$, making the order effect sensitive to both the agent's expressed opinion and the structure of the agent's cognitive ambivalence.
When $\theta = \pi/2$, Eq.~\eqref{eq:DeltaP_formula} simplifies to
\begin{equation}
\Delta P_i^{(1)} = \tfrac{1}{4}(x_i - c_i\cos\phi_i)\,, \quad \Delta P_i^{(2)} = \tfrac{1}{2}x_i\,.
\end{equation}
The quantum coherence term vanishes from $\Delta P_i^{(2)}$ at maximal incompatibility because a non-selective measurement of $Q_B$ preserves $r_i^x = c_i\cos\phi_i$, so this component contributes equally to both orderings and cancels from the asymmetry.

Fig.~\ref{fig:DeltaP} shows $\Delta P_i^{(1)}$ for all agents by community at $\theta = \pi/2$.
At early times, the trajectories fluctuate and span both signs within each community, reflecting the large initial coherence $c_i(0) = 1$ and the diversity of initial phases $\phi_i(0)$.
As $c_i(t) = c_i(0)e^{-t/2}$ decays exponentially, the quantum contribution to $\Delta P_i^{(1)}$ diminishes, and the trajectories converge to a residual determined solely by $x_i(\infty)$.
At steady state, communities with a strong opinion bias, such as $C_1$ and $C_2$, cluster at negative residuals while $C_5$ clusters at positive residuals.
\begin{figure}[htp]
    \centering
    \includegraphics[height=0.25\linewidth]{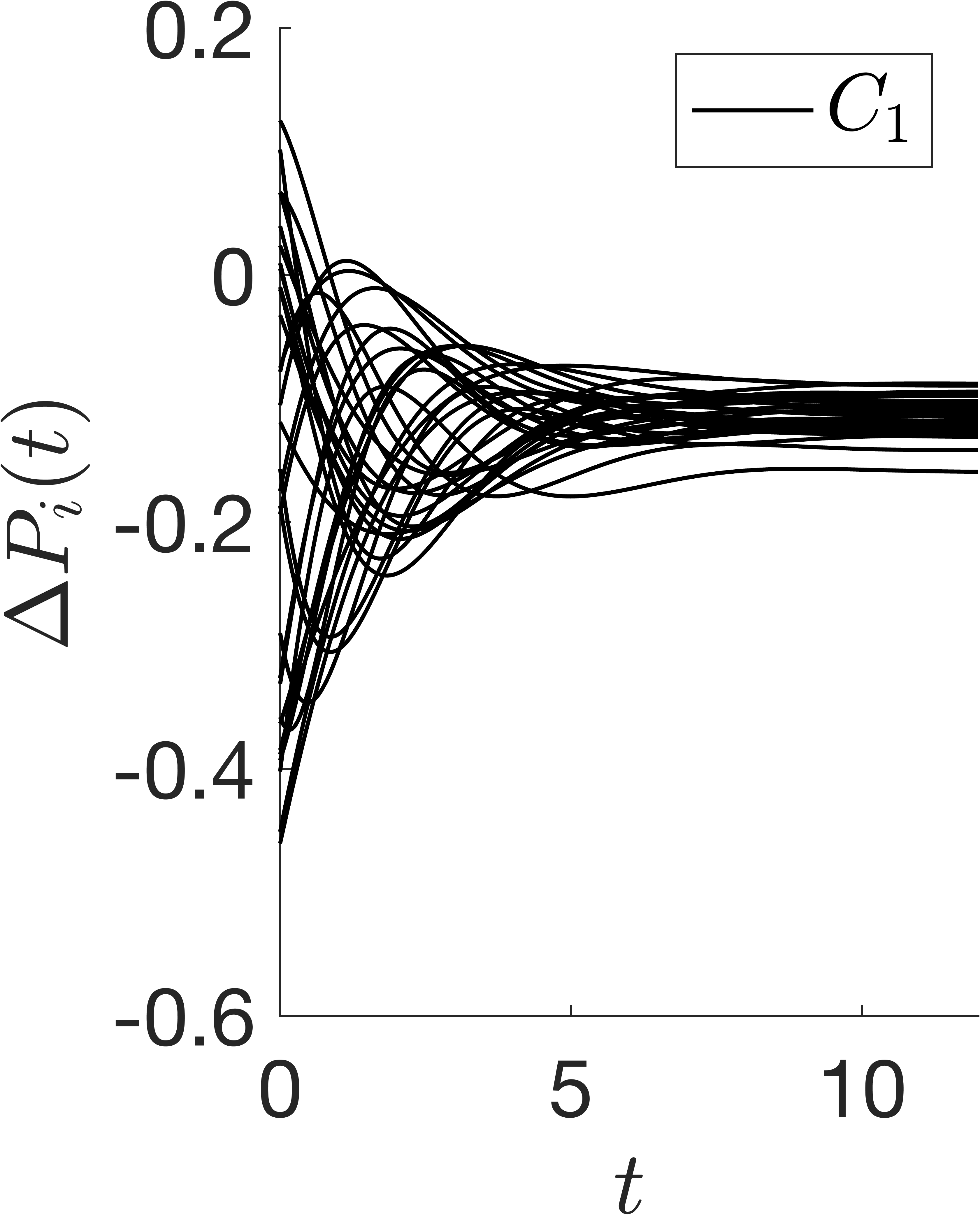}
    \includegraphics[trim={22 0 0 0},clip,height=0.25\linewidth]{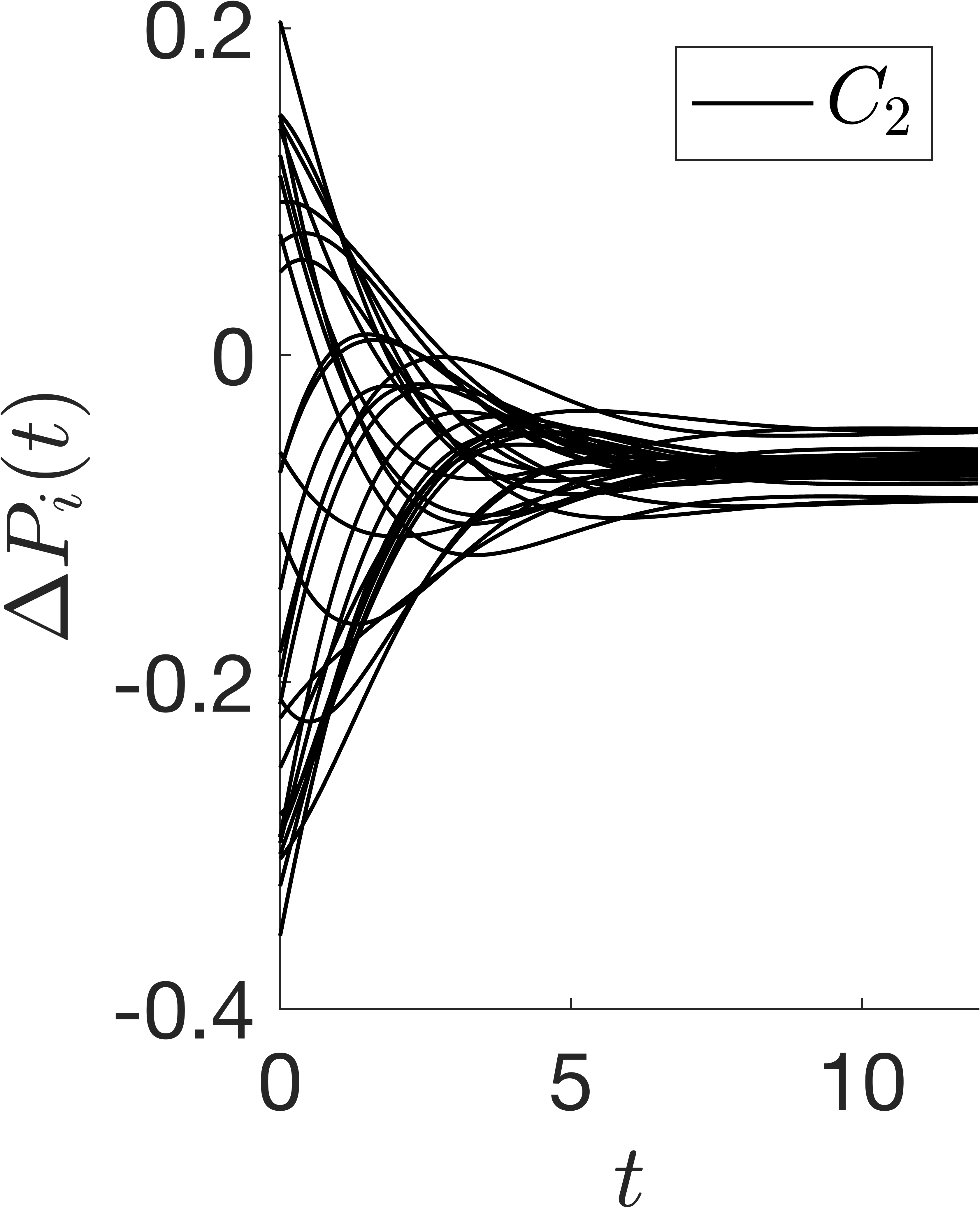}
    \includegraphics[trim={22 0 0 0},clip,height=0.25\linewidth]{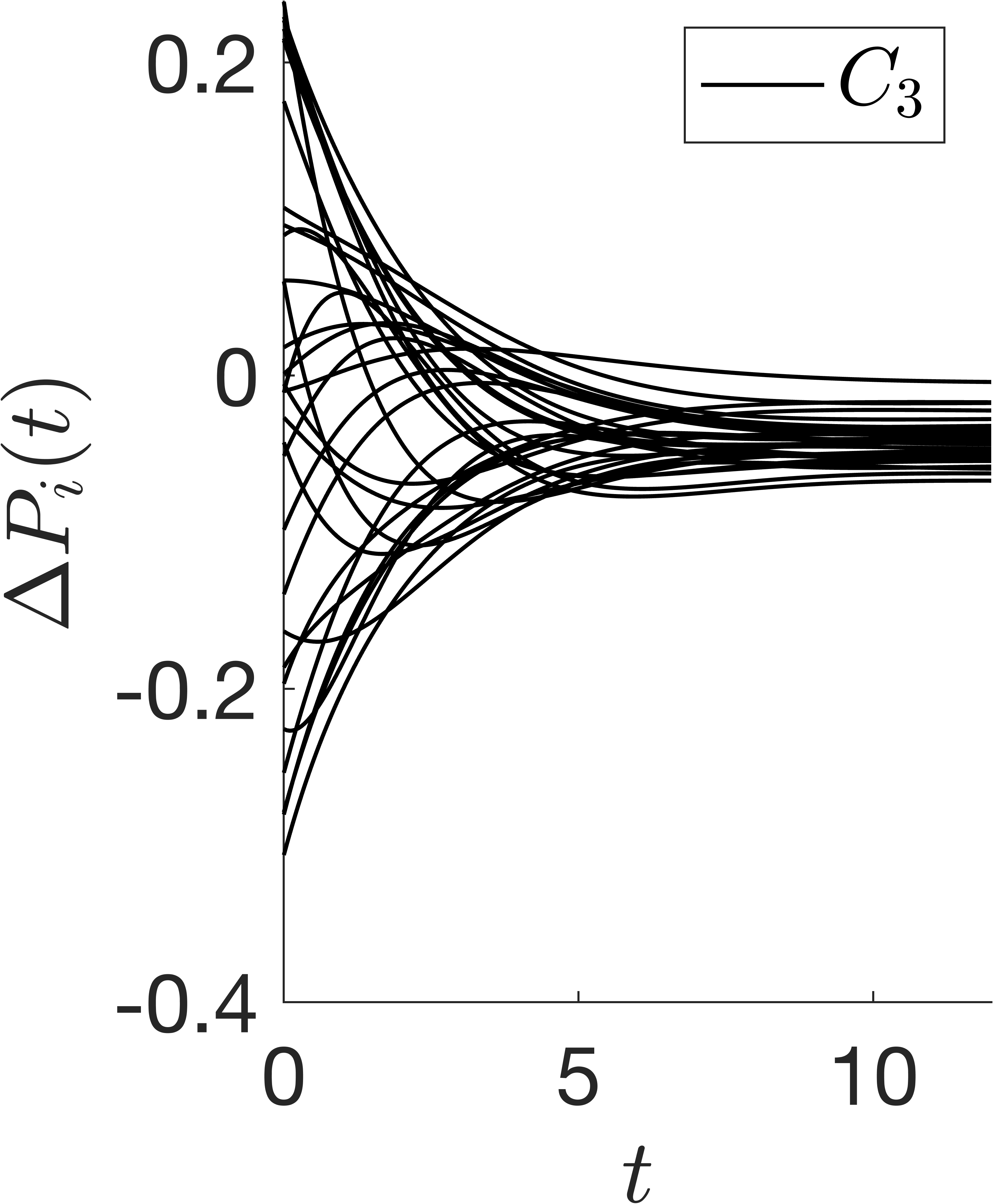}
    \includegraphics[trim={22 0 0 0},clip,height=0.25\linewidth]{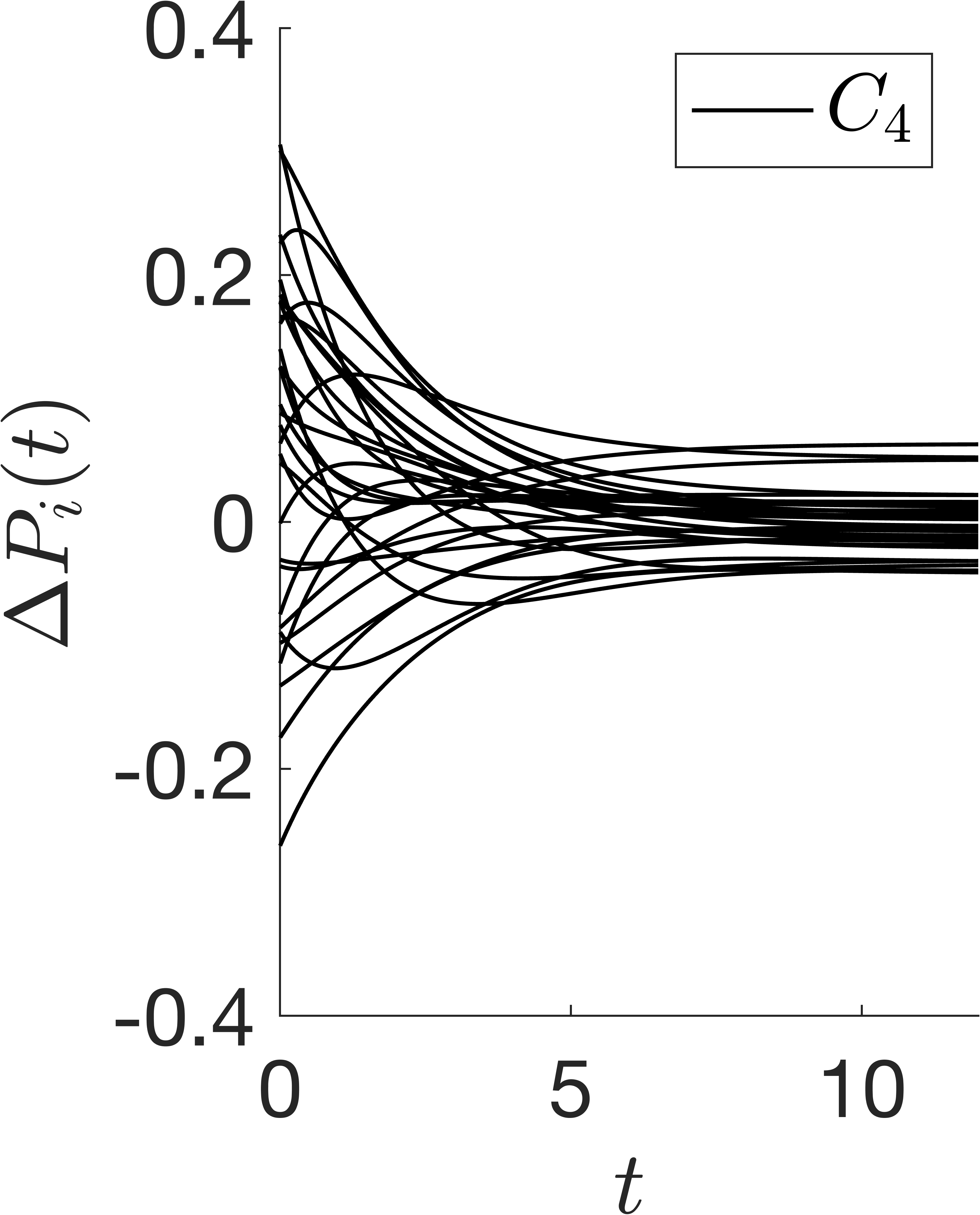}
    \includegraphics[trim={22 0 0 0},clip,height=0.25\linewidth]{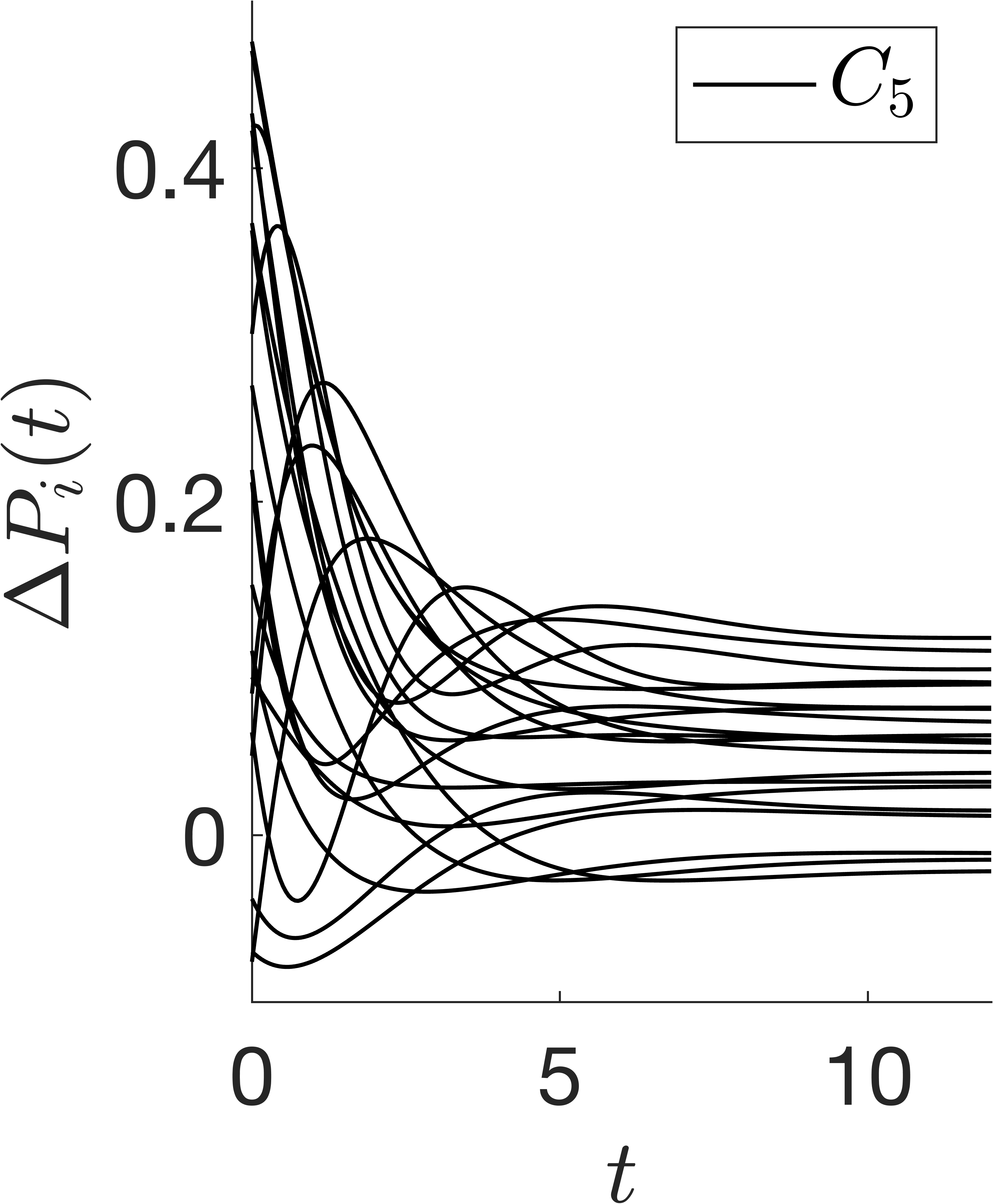}
    \caption{Order-effect asymmetry $\Delta P_i^{(1)}(t)$ at $\theta = \pi/2$ for all agents, grouped by community. Trajectories converge from large transient fluctuations to community-dependent steady-state residuals as coherence magnitude $c_i$ decays.}
    \label{fig:DeltaP}
\end{figure}

Fig.~\ref{fig:DP_distribution} shows the population-level distribution and magnitude of $\Delta P_i^{(1)}$.
The left panel shows the empirical density of $\Delta P_i^{(1)}$ across all agents at several time snapshots for $\theta = \pi/2$.
At $t=0$, the distribution is broad and nearly symmetric around zero, reflecting the large initial coherence $c_i(0) = 1$ and uniformly random phases $\phi_i(0)$.
As time progresses, the distribution narrows and its center shifts slightly negative, consistent with the negative mean opinion of the network.
By $t = 12$, the distribution is sharply concentrated, indicating that the order effect has lost its quantum contribution and is governed entirely by the residual classical opinions $x_i(\infty)$.
The right panel shows the population-averaged absolute order effect $\bar{P}(t;\theta) = \frac{1}{n}\sum_i |\Delta P_i^{(1)}(t;\theta)|$ for various values of $\theta$.
All curves decay monotonically as coherence is lost, consistent with the exponential decay $c_i(t) = c_i(0)e^{-t/2}$ which drives the transient order effect.
Larger values of $\theta$ yield larger initial $\bar{P}$, since both the opinion and coherence coefficients in~\eqref{eq:DeltaP_formula} grow with $\theta$ for $\theta \in (0, \pi/2]$.
For $\theta > \pi/2$ (dashed lines), the curves are smaller and decay to lower residuals, as the opinion coefficient $\frac{\sin^2\theta}{4}$ decreases for $\theta \in (\pi/2, \pi)$.
At steady state, all curves converge to nonzero residuals determined by $\frac{\sin^2\theta}{4}|x_i(\infty)|$, with $\theta$ near $\pi/2$ retaining the largest residual order effect.
Since $\sin^2\theta = \sin^2(\pi - \theta)$, pairs of angles symmetric about $\pi/2$ share the same opinion coefficient and therefore converge to the same steady-state value, explaining the overlapping curves observed at large $t$ for $\theta$ and $\pi - \theta$ in Fig.~\ref{fig:DP_distribution}.
\begin{figure}[htp]
    \centering
    \includegraphics[width=0.28\linewidth]{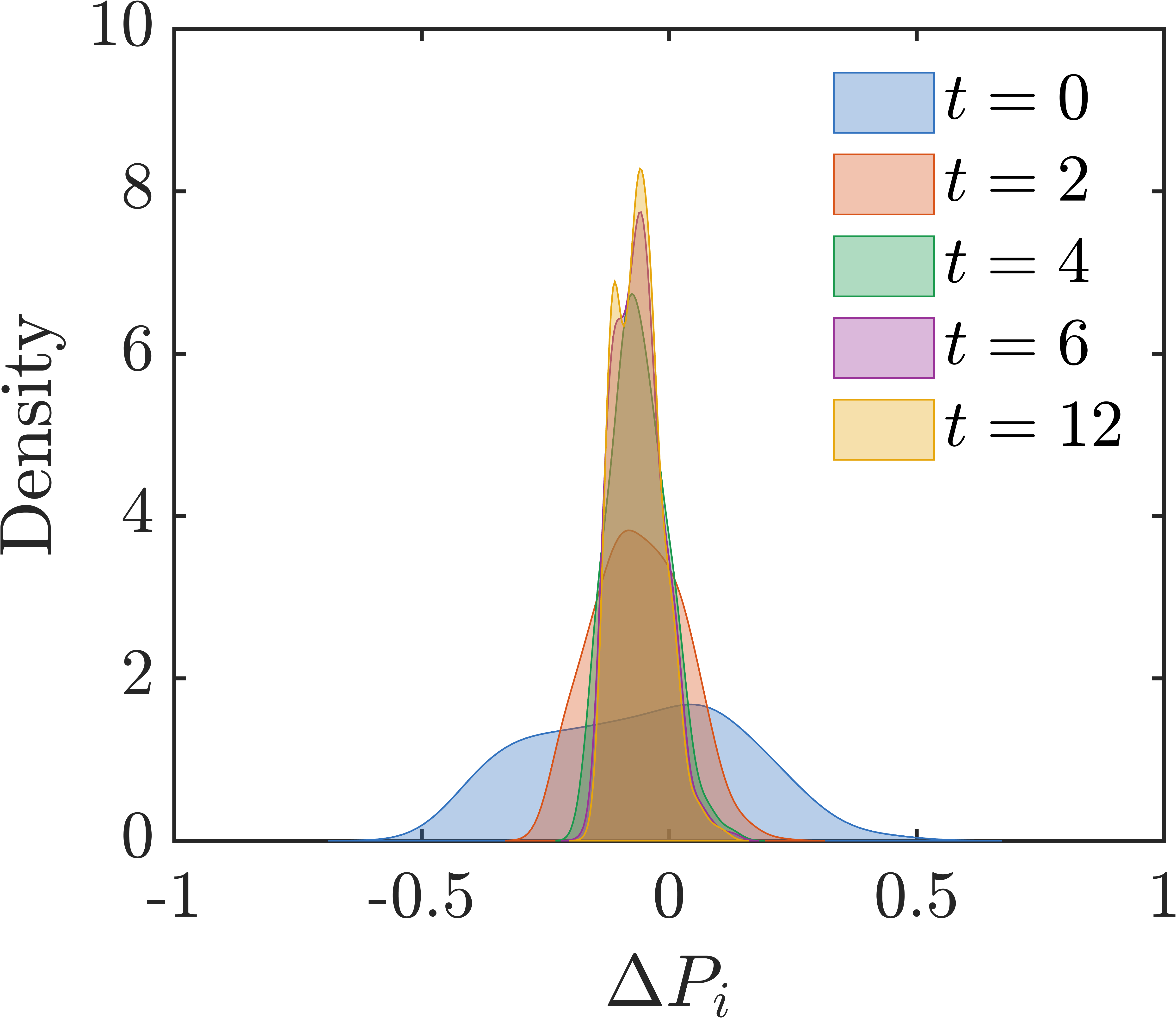} \qquad  \qquad 
    \includegraphics[width=0.28\linewidth]{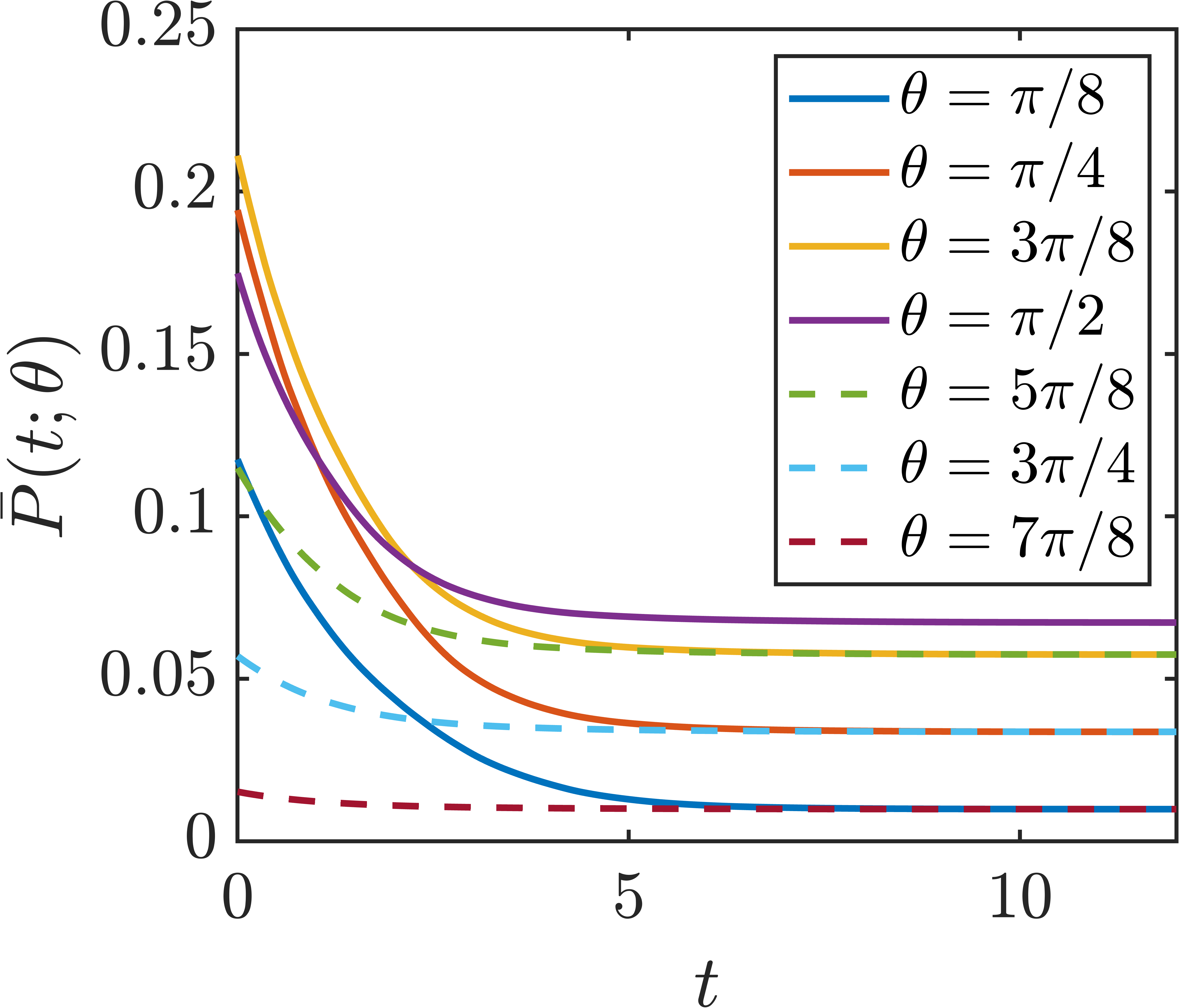}
    \caption{(left): empirical density of $\Delta P_i^{(1)}$ across all agents at $\theta = \pi/2$ for several time snapshots. (right): averaged absolute order effect $\bar{P}(t;\theta)$ for different values of $\theta$.}
    \label{fig:DP_distribution}
\end{figure}

\section{Conclusions and discussion}
\label{sec:conclusion}

We proposed a quantum model of opinion dynamics governed by the Lindblad master equation, in which each agent's cognitive state is a density matrix and social interactions are modeled by dissipative jump operators.
Opinions arise as expectation values of the opinion observable $\hat{O}$, and order effects follow directly from the non-commutativity of measurement operators.
The diagonal of the joint density matrix satisfies an exact continuous-time Markov chain, from which pairwise opinion covariances are defined.
Under the product state approximation, the reduced single-agent equation recovers the classical FJ model~\cite{friedkin1990social}, while the coherence magnitude decays exponentially at a rate independent of the network.

We tested the model on synthetic and real-world networks and observed that the opinion trajectories of the quantum and classical models follow similar trends, while the deviations depend on the network structure, being largest for sparse networks (chain and ring) and smallest for the complete graph.
We have shown that the order effects decay with quantum coherence and converge to a steady state determined by the equilibrium opinion. 
We also examined inter-agent correlations, whose transient dynamics are network-dependent but converge to the same steady state for undirected networks, which is a consequence of the symmetry of the jump operators.

In this paper, we used a Lindblad master equation to model opinion dynamics on social networks, representing each agent's cognitive state as a density matrix and modeling social interactions via dissipative jump operators. 
There are a variety of interesting ways to extend our investigation.
We modeled each agent as a single qubit, which restricts the opinion space to two outcomes. 
It is natural to extend this framework to qudits or higher-dimensional quantum systems, which would allow for a richer representation of opinion states and facilitate empirical calibration against survey data.
In our model, the Lindblad equation is linear, which precludes nonlinear opinion models such as bounded-confidence models. 
One can introduce state-dependent jump operators or interaction rates to yield a nonlinear master equation.
In our framework, the master equation is Markovian. 
Extending to non-Markovian master equations would capture memory effects in social influence, where the impact of past interactions persists over time.
We tested our model on a small set of networks. It is important to conduct more extensive numerical experiments on a broader range of network structures, such as community-based random graphs, scale-free networks, and larger real-world networks, to better understand the dependence of quantum effects on networks.
Finally, empirical calibration of the coherence parameters and anchoring strengths from survey or experimental data remains an important open direction.

\section*{Acknowledgments}
The author acknowledges support from the National Science
Foundation Grants DMS-2350325 and DMS-2514053.

\appendix

\section{Derivation of the reduced master equation for the single-agent density operators}
\label{app:marginal}

In this subsection, we derive the governing equation for the single-agent density operator $\rho_i$ under the product state approximation~\eqref{eq:product_ansatz}.
We apply $\mathrm{Tr}_{\neq i}$ to~\eqref{eq:master_full} and treat the Hamiltonian, the pairwise dissipators, and the anchor dissipators in turn.

Recall the Hamiltonian is $\hat{H} = -\sum_{(k,\ell)\in E} w_{k\ell}\hat{O}_k\hat{O}_\ell$.
For any term with $k \neq i$ and $\ell \neq i$, the commutator $[\hat{O}_k\hat{O}_\ell, \rho]$ acts entirely outside $\mathcal{H}_i$ and its partial trace vanishes.
For terms involving site $i$, we apply $\mathrm{Tr}_{\neq i}$ to $[\hat{O}_i\hat{O}_j, \rho]$ and use the product state approximation $\rho_{ij} \approx \rho_i \otimes \rho_j$, which yields
\begin{equation}
\mathrm{Tr}_{\neq i}\bigl([\hat{O}_i\hat{O}_j, \rho]\bigr)
= \bigl[\hat{O}_i\,\mathrm{Tr}(\rho_j\hat{O}_j),\,\rho_i\bigr]
= x_j\,[\hat{O}_i, \rho_i]\,.
\end{equation}
We sum over all neighbors $j$ of $i$ and obtain
\begin{equation}
\mathrm{Tr}_{\neq i}\bigl(-i[\hat{H},\rho]\bigr)
= -i\bigl[\hat{H}_i^{\mathrm{eff}}, \rho_i\bigr]\,,
\label{eq:app_ham}
\end{equation}
with $\hat{H}_i^{\mathrm{eff}}$ defined in Eq.~\eqref{eq:Heff}.

We next compute the partial trace of the pairwise interaction dissipators.
We consider the term $(1-\lambda_i)w_{ij}\mathcal{D}[\hat{L}_{ij}^+][\rho]$ with $\hat{L}_{ij}^+ = \hat{\sigma}_i^+\hat{\sigma}_j^-$ and apply $\mathrm{Tr}_{\neq i}$ to each part of the dissipator separately.
For the jump part, we use the product state approximation and the cyclic property of the trace and obtain
\begin{equation}
\begin{aligned}
\mathrm{Tr}_{\neq i}\bigl(\hat{L}_{ij}^+\rho(\hat{L}_{ij}^+)^\dagger\bigr)
&= \mathrm{Tr}_{\neq i}\bigl(\hat{\sigma}_i^+\hat{\sigma}_j^-(\rho_i\otimes\rho_j)\hat{\sigma}_i^-\hat{\sigma}_j^+\bigr) \\
&= \hat{\sigma}_i^+\rho_i\hat{\sigma}_i^-\mathrm{Tr}(\hat{\sigma}_j^-\rho_j\hat{\sigma}_j^+) \\
&= \mathrm{Tr}(\rho_j\hat{\sigma}_j^+\hat{\sigma}_j^-)\,\hat{\sigma}_i^+\rho_i\hat{\sigma}_i^-\,.
\end{aligned}
\end{equation}
Since $\hat{\sigma}_j^+\hat{\sigma}_j^- = \ketbra{0}{0}_j$, we have $\mathrm{Tr}(\rho_j\hat{\sigma}_j^+\hat{\sigma}_j^-) = \frac{1+x_j}{2}$.
Therefore, the anticommutator part yields
\begin{equation}
\begin{aligned}
-\tfrac{1}{2}\mathrm{Tr}_{\neq i}\bigl(\{(\hat{L}_{ij}^+)^\dagger\hat{L}_{ij}^+, \rho\}\bigr)
= -\tfrac{1}{2}\mathrm{Tr}(\rho_j\hat{\sigma}_j^+\hat{\sigma}_j^-)\,\{\hat{\sigma}_i^-\hat{\sigma}_i^+, \rho_i\}  = -\frac{1+x_j}{4}\,\{\hat{\sigma}_i^-\hat{\sigma}_i^+, \rho_i\}\,.
\end{aligned}
\end{equation}
We combine the jump and anticommutator terms and obtain
\begin{equation}
\mathrm{Tr}_{\neq i}\bigl(\mathcal{D}[\hat{L}_{ij}^+][\rho]\bigr)
= \frac{1+x_j}{2}\,\mathcal{D}[\hat{\sigma}_i^+][\rho_i]\,.
\label{eq:app_social_plus}
\end{equation}
We apply the identical calculation to $\hat{L}_{ij}^- = \hat{\sigma}_i^-\hat{\sigma}_j^+$, using $\hat{\sigma}_j^-\hat{\sigma}_j^+ = \ketbra{1}{1}_j$ and $\mathrm{Tr}(\rho_j\hat{\sigma}_j^-\hat{\sigma}_j^+) = \frac{1-x_j}{2}$, and obtain
\begin{equation}
\mathrm{Tr}_{\neq i}\bigl(\mathcal{D}[\hat{L}_{ij}^-][\rho]\bigr)
= \frac{1-x_j}{2}\,\mathcal{D}[\hat{\sigma}_i^-][\rho_i]\,.
\label{eq:app_social_minus}
\end{equation}

For the anchor dissipators, since $\hat{\sigma}_i^\pm$ act only on $\mathcal{H}_i$, we have
\begin{equation}
\mathrm{Tr}_{\neq i}\bigl(\mathcal{D}[\hat{\sigma}_i^\pm][\rho]\bigr)
= \mathcal{D}[\hat{\sigma}_i^\pm][\rho_i]\,.
\label{eq:app_anchor}
\end{equation}

Combining Eqs.~\eqref{eq:master_full}, \eqref{eq:app_ham}, \eqref{eq:app_social_plus}, \eqref{eq:app_social_minus} and \eqref{eq:app_anchor}, we obtain the governing equation of $\rho_i$ in~\eqref{eq:master_marginal}.

\section{Derivation of the quantum coherence dynamics}
\label{app:coherence}
We derive governing equations for the quantum coherence magnitude $c_i$ and phase $\phi_i$ in this subsection. 

Recall that the off-diagonal element $\rho_i^{01} = \bra{0}\rho_i\ket{1} = \frac{c_i}{2}e^{-i\phi_i}$ and satisfies the master equation~\eqref{eq:master_marginal}. 
Since $\hat{H}_i^{\mathrm{eff}} = -\Omega_i\hat{O}_i$ in \eqref{eq:Heff} and $\hat{O}_i$ is diagonal in the computational basis, its matrix elements are $h_{00} = \bra{0}\hat{H}_i^{\mathrm{eff}}\ket{0} = -\Omega_i$ and $h_{11} = \bra{1}\hat{H}_i^{\mathrm{eff}}\ket{1} = \Omega_i$.
We apply $\bra{0}\cdot\ket{1}$ to the Hamiltonian term and obtain
\begin{equation}
\bra{0} -i[\hat{H}_i^{\mathrm{eff}},\rho_i]\ket{1} = -i(h_{00} - h_{11})\rho_i^{01} = 2i\Omega_i\,\rho_i^{01}\,.
\label{eq:rho01_H}
\end{equation}

We next compute the dissipator contributions.
For each operator in~\eqref{eq:master_marginal}, we evaluate $\bra{0}\mathcal{D}[\hat{L}][\rho_i]\ket{1}$ directly. 
For $\hat{L} = \hat{\sigma}_i^+$, the jump part vanishes since
\begin{equation}
\bra{0}\hat{\sigma}_i^+\rho_i\hat{\sigma}_i^-\ket{1}
= \bra{0}\ket{0}\bra{1}\rho_i\ket{1}\braket{0}{1} = 0\,.
\end{equation}
For the anticommutator part, notice that $(\hat{\sigma}_i^+)^\dagger\hat{\sigma}_i^+ = \hat{\sigma}_i^-\hat{\sigma}_i^+ = \ketbra{1}{1}$, which gives $\bra{0}\hat{\sigma}_i^-\hat{\sigma}_i^+\ket{0} = 0$ and $\bra{1}\hat{\sigma}_i^-\hat{\sigma}_i^+\ket{1} = 1$. Therefore, we have
\begin{equation}
\begin{aligned}
-\tfrac{1}{2}\bra{0}\{(\hat{\sigma}_i^+)^\dagger\hat{\sigma}_i^+, \rho_i\}\ket{1} = -\tfrac{1}{2}\rho_i^{01}\,,
\end{aligned}
\end{equation}
yielding $\bra{0}\mathcal{D}[\hat{\sigma}_i^+][\rho_i]\ket{1} = -\frac{1}{2}\rho_i^{01}$.
The same calculation for $\hat{L} = \hat{\sigma}_i^-$ gives $\bra{0}\mathcal{D}[\hat{\sigma}_i^-][\rho_i]\ket{1} = -\frac{1}{2}\rho_i^{01}$, and every jump operator contributes $-\frac{1}{2}\rho_i^{01}$ per unit rate. 

We sum all dissipator terms in~\eqref{eq:master_marginal} and use the row-stochasticity of $W$, obtaining
\begin{equation}
\dot{\rho}_i^{01}\big|_{\mathcal{D}} = -\tfrac{1}{2}\rho_i^{01}\,.
\label{eq:rho01_D}
\end{equation}
Combining~\eqref{eq:rho01_H} and~\eqref{eq:rho01_D}, we obtain the governing equation for $\rho_i^{01}$
\begin{equation}
\dot{\rho}_i^{01} = \Bigl(2i\Omega_i - \tfrac{1}{2}\Bigr)\rho_i^{01}\,,
\label{eq:rho01_full}
\end{equation}
which yields the real and imaginary parts 
\begin{equation} \label{eq:coherence_dynamics}
\begin{aligned}
\dot{c}_i &= -\tfrac{1}{2}\,c_i\,,  \\
\dot{\phi}_i &= -2\sum_j w_{ij}\,x_j\,.
\end{aligned}
\end{equation} 

\bibliographystyle{abbrv}
\bibliography{references}

\end{document}